\documentclass[10pt]{article}
\usepackage{pgfplots}
\pgfplotsset{compat=1.14}
\usepackage{tikz-3dplot}
\usepackage{xspace}
\usepackage{tikz}
\usepackage{multirow}
\usepackage{marvosym}
\usepackage{enumitem}
\usepackage{fullpage}
\usepackage{amsthm}
\usepackage{amssymb}
\usepackage{mathtools}
\usepackage{amsfonts}
\usepackage{booktabs}
\usepackage{subcaption}
\usepackage{relsize}
\usepackage[T1]{fontenc}
\usepackage{makecell}

\usepackage{hyperref}
\pgfplotsset{compat=1.16}
\usepackage{tikz-3dplot}
\usepackage{xspace}
\usepackage{tikz}
\pgfdeclarelayer{background}
\pgfsetlayers{background,main}

\usepackage{xcolor,colortbl}

\usepackage{thmtools, thm-restate}

\usepackage[colorinlistoftodos]{todonotes}
\usepackage{bm}
\usepackage{amsthm}      
\usepackage{amsmath}     
\usepackage{algorithm}   
\usepackage{algpseudocode} 
\usepackage{booktabs}
\usepackage{xspace}
\usepackage{etoolbox}

\usetikzlibrary{shapes.geometric}

\newcommand{\vars}{\textsf{vars}}
\newcommand{\free}{\textsf{free}}
\newcommand{\bound}{\textsf{bound}}
\newcommand{\out}{\textsf{OUT}}
\newcommand{\bigO}[1]{\mathcal{O}(#1)}
\newcommand{\cw}{\mathsf{cw}}
\newcommand{\ctw}{\mathsf{cw}}

\newcommand{\cg}{\mathsf{cg}}
\newcommand{\eliminate}{\mathsf{eliminate}}
\newcommand{\rev}{\texttt{rev}}

\newcommand{\fw}{\mathsf{w}}

\newcommand{\atoms}{\textsf{atoms}}
\newcommand{\ftd}{\text{\textsf FTD}}
\newcommand{\bags}{\text{\textsf bags}}

\newcommand{\ospgosyan}{\textsf{OSPG}+\textsf{OSYan}\xspace}
\newcommand{\Dom}{\text{\sf Dom}}

\newcommand{\defeq}{\stackrel{\text{def}}{=}}

\newcommand{\calT}{\mathcal{T}}
\newcommand{\calR}{\mathcal{R}}
\newcommand{\calB}{\mathcal{B}}
\newcommand{\calA}{\mathcal{A}}
\newcommand{\calF}{\mathcal{F}}

\newcommand{\calS}{\mathcal{S}}

\newcommand{\pg}{\textsf{PG}\xspace}
\newcommand{\ospg}{\textsf{OSPG}\xspace}

\newcounter{magicrownumbers}
\newcommand\rownumber{\footnotesize\stepcounter{magicrownumbers}\arabic{magicrownumbers}}
\newcommand{\linenumber}{\makebox[2ex][r]{\rownumber\TAB}}

\newcommand{\TAB}{\makebox[2.5ex][r]{}}

\newcommand{\LET}{\textbf{let}\xspace}%
\newcommand{\IF}{\textbf{if}\xspace}%
\newcommand{\ELSE}{\textbf{else}\xspace}%
\newcommand{\FOREACH}{\textbf{foreach}\xspace}%
\newcommand{\RETURN}{\textbf{return}\xspace}%

\newcommand{\nop}[1]{}
\definecolor{verylightgray}{gray}{0.90}
\newcommand{\mahmoud}[1]{\todo[inline,color=cyan]{\textsf{#1} \hfill \textsc{--Mahmoud.}}}

\newcommand{\dan}[1]{\todo[inline,color=green]{\textsf{#1} \hfill \textsc{--Dan.}}}

\newcommand{\dans}[1]{\todo[inline,color=pink]{\textsf{#1} \hfill \textsc{--DanS.}}}

\definecolor{verylightgray}{gray}{0.90}
\newcommand{\ahmet}[1]{\todo[inline,color=verylightgray]{\textsf{#1} \hfill \textsc{--Ahmet.}}}

\newcommand{\alex}[1]{\todo[inline,color=orange]{\textsf{#1} \hfill \textsc{--Alexandru.}}}

\nop{
\newcommand{\mahmoud}[1]{}
\newcommand{\dan}[1]{}
\newcommand{\dans}[1]{}
\newcommand{\ahmet}[1]{}
\newcommand{\alex}[1]{}
}

\theoremstyle{plain}                  
\newtheorem{theorem}{Theorem}

\newtheorem{proposition}[theorem]{Proposition}

\newtheorem{definition}[theorem]{Definition}
\newtheorem{example}[theorem]{Example}

\newtheorem{claim}[theorem]{Claim}

\date{}
\title{Output-Sensitive Evaluation of \\ Acyclic Conjunctive Regular Path Queries}

\author{
\hspace{-1.5em}
\begin{tabular}{ccc}
Mahmoud Abo Khamis & \hspace{0.2cm} Alexandru-Mihai Hurjui & \hspace{0.2cm} Ahmet Kara \\
Relational AI & University of Zurich & OTH Regensburg \\
\texttt{\normalsize mahmoud.abokhamis@relational.ai} & \texttt{\normalsize hurjuialexandru12@gmail.com}   & \texttt{\normalsize ahmet.kara@oth-regensburg.de} \\[0.6cm]
 Dan Olteanu & Dan Suciu & Zilu Tian  \\
 University of Zurich & University of Washington & OmniVision \\
\texttt{\normalsize dan.olteanu@uzh.ch}  & \texttt{\normalsize suciu@cs.washington.edu}  & \texttt{\normalsize ruby.tian@ovt.com}
\end{tabular}
}

\begin{document}
\maketitle

\begin{abstract}
Conjunctive Regular Path Queries, or CRPQs for short, are an essential construct in graph query languages. In this paper, we propose the first output-sensitive algorithm for evaluating acyclic CRPQs. It is output-sensitive in the sense that its complexity is a function of the sizes of the input graph and of the query output. In particular, it does not depend on the output sizes of the regular expressions that appear in the query, as these sizes can be much larger than the query output size.

Our algorithm proceeds in two stages. In the first stage, it contracts the given query into a free-connex acyclic one such that the output of the original query can be obtained from the output of the contracted one. This contraction removes bound variables by composing regular expressions or by promoting bound variables to free ones. The minimum necessary number of promoted bound variables gives the contraction width, which is a novel parameter specific to CRPQs. In the second stage, our algorithm evaluates the free-connex acyclic CRPQ and projects away the columns of the promoted bound variables. It ensures output-sensitivity by computing the calibrated outputs of the regular expressions appearing in the free-connex acyclic CRPQ in time proportional to their sizes.

Our algorithm has lower complexity than the state-of-the-art approaches for problem instances where (i) the query output is asymptotically smaller than the worst-case output size or (ii) the largest output size of any of the regular expression in the query.
\end{abstract}

\paragraph{Acknowledgements}
Suciu was partially supported by NSF IIS 2314527, NSF SHF 2312195, and NSF IIS 2507117.

\section{Introduction}
\label{sec:intro}

The evaluation of conjunctive regular path queries is a central problem in graph databases. Such queries are an essential construct in graph query languages in both academia and industry, e.g.,~\cite{SPARQL:2013,Cypher:SIGMOD:2018,GSQL:2021,PGQL:2021,DBLP:conf/sigmod/DeutschFGHLLLMM22,DBLP:conf/icdt/FrancisGGLMMMPR23}. Regular path queries have been investigated extensively, e.g.,~\cite{MendelzonW95,CalvaneseGLV03,Wood12,Baeza13,AnglesABHRV17}, as they represent the key differentiator of graph query languages from the purely relational query languages. 

Let a graph $G=(V,E,\Sigma)$, where $V$ is the set of vertices and $E$ is the set of edges that are labeled with symbols from an alphabet $\Sigma$. A regular path query, or RPQ for short, $r$ is a regular expression over $\Sigma$. Its semantics is given by the set of vertex pairs $(v, u)$ such that $u$ can be reached in $G$ from $v$ via a path labeled with a string from the language $L(r)$. A conjunctive regular path query, or CRPQ for short, $Q$ is a conjunction of binary atoms, with each atom  defined by a regular path query: $Q(\bm Y) = \wedge_{i\in[n]} r_i(\bm X_i)$, where each $r_i$ is an RPQ, each $\bm X_i$ is a pair of variables, and  $\bm Y\subseteq \bigcup_{i\in[n]} \bm X_i$ is the set of free variables.

In this paper, {\em we propose the first output-sensitive algorithm for evaluating acyclic\footnote{Since the CRPQs are conjunctions of binary atoms, the various notions of acyclicity become equivalent: $\alpha$-acyclicity is the same as Berge-acyclicity. We therefore simply refer to acyclic CRPQs in this paper.} CRPQs}. It is output-sensitive in the sense that its complexity is a function of the sizes of the input graph and of the query output. In particular, it does not depend on the output sizes of the RPQs in the query, as these sizes can be larger than the query output size.

\begin{restatable}{theorem}{MainTheorem}
\label{theo:general}
    Any acyclic CRPQ $Q$ over an edge-labeled graph $G = (V,E, \Sigma)$ can be evaluated in $\bigO{|E| + |E|\cdot\out^{1/2} \cdot |V|^{\ctw/2} + \out\cdot |V|^{\ctw}}$ data complexity, where $\out$ is the size of the output of $Q$ and $\ctw$ is the contraction width of $Q$.
\end{restatable}

In a nutshell, our algorithm proceeds in two stages. In the first stage, it contracts the given CRPQ $Q$ into a free-connex CRPQ $Q'$ such that the output of $Q$ can be obtained from the output of $Q'$. This contraction repeatedly uses two mechanisms to rewrite the query until 
they cannot be applied anymore.
The first mechanism  removes a bound variable that appears in at most two atoms of $Q$. This corresponds to concatenating the RPQs in the two atoms into a single RPQ. It exploits the closure of RPQs under composition and does not increase the data complexity of RPQ evaluation. 
The second mechanism promotes those bound variables that violate the free-connex property to free variables so that the query becomes free-connex. This can increase the query output size by a multiplicative factor of at most $|V|$ for each promoted bound variable. The minimum necessary number of promoted bound variables gives the contraction width $\ctw$ of $Q$, which is a novel parameter specific to CRPQs. The size of the output of $Q'$ can thus be at most $\out\cdot|V|^{\ctw}$, where $\out$ is the size of the output of $Q$.

In the second stage, our algorithm evaluates the free-connex acyclic CRPQ $Q'$. For this, it first computes the {\em calibrated} outputs of the RPQs in $Q'$ in time proportional to their sizes. For an RPQ atom $r(X,Y)$ in $Q'$, where $S\subseteq\{X,Y\}$ is the set of free variables of $Q'$ that occur in this atom, its calibrated output is the projection of the output of $Q'$ onto $S$. This calibration ensures that the time to compute the individual RPQs does not exceed the size of the output of $Q'$. It is reminiscent of the classical Yannakakis algorithm~\cite{Yannakakis81}: It is done in one bottom-up pass followed by one top-down pass over the join tree of $Q'$. However, whereas the Yannakakis algorithm works on input relations that are materialized and have sizes linear in the input database size, our algorithm works with non-materialized RPQs whose output sizes can be larger than 
the input graph size and even than the query output size. In our setting, calibration is thus essential for output sensitivity.
By joining the calibrated RPQ outputs, the algorithm then computes the output of $Q'$. The output of $Q$ is then computed by projecting away the columns of the promoted bound variables from the output of $Q'$.

\nop{
\ahmet{The above description of our algorithm is bit different from what we actually do in the main body. Maybe, this was intended to keep the description high-level. I describe in the following what we actually do, so we are on the same page (we do not need to be that precise here of course).\\
We construct a free-connex tree decomposition for the given CRPQ $Q$. For each bag, we create a query $Q'$ induced by that bag. Then, for each induced query $Q'$, we do the following. We eliminate bound variables and adapt, if necessary, the input graph to preserve the output of 
$Q'$. This process can produce new non-join bound variables. Then, we promote those bound variables that are in $Q'$, but which we are not able to eliminate, to free variables. This gives us a free-connex query $Q''$ resulting from $Q'$. Then we compute the output of $Q''$ and project the result onto the free variables of $Q'$. Finally, we join the outputs obtained for each induced query $Q'$. The contraction width is the maximal number of promoted variables over all bags.
}
}

\nop{
\dans{this definition seem vacuous: if the output is non-empty, then for every output tuple $(x,y,\ldots)$ each of the atoms (pairs?) must be satisfied, hence every pair contributes to that output tuple. I believe you meant to say something like this: ``for every query atom $R(X,Y)$ and any pair of nodes $(x,y)$ satisfying the regular expression $R$, there exists a query output that uses the pair $(x,y)$.''}.
}

Although the contraction width measures the complexity of transforming acyclic CRPQs into free-connex acyclic CRPQs, it differs from the free-connex fractional hypertree width for conjunctive queries~\cite{OlteanuZ15,Hu25}. 

\begin{restatable}{proposition}{PropertiesContractionWidth}
\label{prop:properties_contraction_width}
There is an infinite class of CRPQs such that the contraction  width of the queries in the class is bounded but their free-connex fractional hypertree width is unbounded.   
\end{restatable}

For instance, the free-connex acyclic CRPQs have free-connex fractional hypertree width $\fw=1$ and contraction width $\ctw=0$. The $k$-path CRPQs, where the join variables are free and the other variables are bound, have $\fw=2$ and $\ctw=0$. The star CRPQs, where the join variable is bound and all other variables are free, have unbounded free-connex fractional hypertree width and $\ctw=1$.

Of particular interest are CRPQs with contraction width 0, since our algorithm evaluates them with the lowest complexity among all CRPQs: $\bigO{|E| + |E|\cdot \out^{1/2}+\out}$. This matches the best known output-sensitive data complexity for evaluating RPQs~\cite{KhamisKOS25}. Examples of CRPQs with contraction width 0 include: (i) free-connex acyclic CRPQs, but also (ii) non-free-connex acyclic CRPQs that are obtained from free-connex acyclic CRPQs by replacing any RPQ atom, whose variables are free, by an arbitrarily long path CRPQ whose free variables are the same as of the replaced atom. For instance, the $k$-path CRPQs with $k\geq 2$ and an arbitrary set of free variables has contraction width 0.

For free-connex acyclic CRPQs, we give a more refined complexity of our algorithm that uses the maximum size $\out_c$ of the calibrated output of any RPQ in the query. Note that by construction, $\out_c\leq\out$. Also, since all RPQ atoms are binary, $\out_c \leq |V|^2$.

\begin{restatable}{proposition}{FreeConnex}
\label{prop:free-connex}
    \begin{itemize}
        \item For any free-connex acyclic CRPQ $Q$, it holds $\ctw(Q) =0$.
        \item     Any free-connex acyclic CRPQ $Q$ over an edge-labeled graph $G = (V,E, \Sigma)$ can be evaluated in $\bigO{|E| + |E|\cdot \out_c^{1/2} + \out}$ data complexity.
    \end{itemize}
\end{restatable}

\nop{
\dans{It's unclear why that's the case.  If you add another ``tree'' (or path?) connecting the endpoints then the query becomes cyclic.}. 

\dan{Indeed, sloppy language. I rephrased locally: The idea is to {\em replace} an edge between two free variables by a tree of bound variables.}
}




\subsection{Comparison with Related Work}

{\em None of the existing evaluation algorithms for CRPQs is output-sensitive}: Their complexity is not only a function of the input data graph and query output sizes but also of the output sizes of the RPQs in the query, yet the RPQ outputs can be asymptotically larger than the query output. In the following, we overview the state-of-the-art algorithms for CRPQ evaluation and show that, for problem instances where the query output is asymptotically smaller than the worst-case output size or even than the largest output size of any of the RPQs in the query, our algorithm can outperform them.

The classical evaluation algorithm for RPQs, called \pg, computes the product graph of the data graph and the non-deterministic finite automaton defined by the given RPQ and runs in  $\bigO{|V|\cdot |E|}$ data complexity~\cite{Baeza13, MartensT18, MendelzonW95}. This cannot be improved by combinatorial algorithms: There is no combinatorial\footnote{This term is not defined precisely. A combinatorial algorithm does not rely on Strassen-like fast matrix multiplication. Using fast matrix multiplication, the data complexity of \pg can be improved to $\bigO{|V|^{\omega}}$, where $\omega$ is the matrix multiplication exponent~\cite{CaselS23} and currently $\omega = 2.37$.} algorithm to compute RPQs with $\bigO{(|V|\cdot |E|)^{1-\epsilon}}$ data complexity\footnote{The lower bound in~\cite{CaselS23} has $|V|+|E|$ instead of $|E|$, but the proof also holds for the formulation that just uses $|E|$.} for any $\epsilon >0$, unless the combinatorial Boolean Matrix Multiplication conjecture fails~\cite{CaselS23}.\footnote{The combinatorial Boolean Matrix Multiplication  conjecture is as follows: Given two $n\times n$ Boolean matrices $A,B$, there is no combinatorial algorithm that multiplies $A$ and $B$ in $\bigO{n^{3-\epsilon}}$ for any $\epsilon > 0$.} When taking the query output into consideration, however, \pg's complexity can be improved. The output-sensitive variant of \pg, called \ospg, evaluates any RPQ with $\bigO{|E| + |E|\cdot \out^{1/2}}$ data complexity, where $\out$ is the output size of the  RPQ~\cite{KhamisKOS25}.\footnote{The data complexity of \ospg is $\bigO{|E| \cdot \Delta + \min(\out/\Delta, |V|) \cdot |E|}$~\cite{KhamisKOS25}. The user-defined threshold $\Delta$ was set to $|E|^{1/2}$ in~\cite{KhamisKOS25}, yielding $\bigO{|E|^{3/2} + \min(\out\cdot |E|^{1/2}, |V|\cdot|E|)}$. For $\Delta = \out^{1/2}$, we obtain instead $\bigO{|E|\cdot \out^{1/2} + |E|}$, where the additive term $|E|$ is necessary to scan the input data even when for $\out=0$.} The complexity of \ospg is at most the complexity of \pg, since $\out \leq |V|^2$.

The state-of-the-art evaluation algorithm for CRPQs proceeds in two stages~\cite{DBLP:conf/icdt/CucumidesRV23}. Using \pg, it first materializes the binary relations representing the output of the RPQs in the given CRPQ $Q$. This stage turns the CRPQ into an equivalent conjunctive query $Q'$ over materialized binary relations. The algorithm then evaluates $Q'$ using a worst-case optimal join algorithm~\cite{Veldhuizen14,Ngo:JACM:18}. 
For acyclic CRPQs with arbitrary free variables, we can immediately improve this state-of-the-art. First, we use \ospg instead of \pg to materialize the RPQs with $\bigO{|E| + |E|\cdot \out_a^{1/2}}$ data complexity, where $\out_a \leq |V|^2$ is the maximum size of the {\em non-calibrated} output of the RPQs.\footnote{$\out_a$ is not to be confused with $\out_c$, which is the maximum size of the calibrated output of  any RPQ in $Q$. It holds that $\out_c\leq \out_a$ and there can be instances for which $\out_c=0$ yet $\out_a=|V|^2$.} Then, $Q'$ can be evaluated using the classical Yannakakis algorithm, called here \textsf{Yan},~\cite{Yannakakis81} with $\bigO{\out_a\cdot\out}$ data complexity, or using its output-sensitive refinement~\cite{Hu25}, called here \textsf{OSYan}, in $\bigO{\out_a\cdot \out^{1-1/\fw(Q')}+\out}$ data complexity, where $\fw(Q')$ is the free-connex fractional hypertree width of $Q'$ and $\out$ is the output of $Q'$ (and also of $Q$). Overall, the best complexity for evaluating acyclic CRPQs using prior work can therefore be achieved by combining \ospg and \textsf{OSYan}: $\bigO{|E| + |E|\cdot\out_a^{1/2}+\out_a\cdot\out^{1-1/\fw(Q')}+\out}$. Yet this approach, as well as all previous approaches, are {\em not} output-sensitive: Indeed, one can easily construct data graphs for which the output of some RPQs in $Q$ consists of all quadratically many pairs of vertices in the data graph, i.e., $\out_a=\bigO{|V|^2}$, while the query output is empty, i.e., $\out=0$. In that case, our algorithm takes time $\bigO{|E|}$, whereas \ospgosyan takes $\bigO{|E|\cdot|V|}$. 

For free-connex acyclic CRPQs, which have free-connex fractional hypertree width 1, \ospgosyan needs  
$\bigO{|E| + |E|\cdot\out_a^{1/2}+\out_a +\out}$, which simplifies to $\bigO{|E| + |E|\cdot\out_a^{1/2} +\out}$, since the time $\bigO{|E|\cdot\out_a^{1/2}}$ to compute the RPQs cannot be smaller than the maximum output size of the RPQs. In contrast, our algorithm needs $\bigO{|E| + |E|\cdot\out_c^{1/2} +\out}$ (Proposition~\ref{prop:complexity_eval_free_connex}).
Since $\out_c \leq \out_a$, this implies that for free-connex acyclic queries, our algorithm performs at least as good as \ospgosyan for any input and outperforms \ospgosyan in case $\out_a$ is larger than $\out_c$. 
The next examples compare our algorithm with \ospgosyan on CRPQs that are not free-connex acyclic.  
Further examples are in Appendix~\ref{subsec:app_query_examples}.

\nop{In case all variables of $Q'$ are free, the running time of the second stage is proportional to the worst-case output size of $Q'$: $\bigO{\out_a^{\rho^*(Q')} + \out}$, where $\out_a \leq |V|^2$ is the maximum size of the {\em non-calibrated} output of the RPQs, $\out$ is the output of $Q'$ (and also of $Q$), and $\rho^*(Q')$ is the fractional edge cover of $Q'$. Note that $\out_a$ is not to be confused with $\out_c$, which is the maximum size of the calibrated output of  any RPQ in $Q$. 
\dans{Are you considering cyclic queries later in the paper?  I just wonder if the discussion of $\rho^*$ here is really needed.}
}

\begin{example}
\rm
    Let the $k$-path CRPQ $Q_1(X_1,X_{k+1}) = r_1(X_1,X_2),\ldots,r_{k}(X_{k},X_{k+1})$, where $r_1,\ldots,r_k$ are RPQs. The contraction width of $Q_1$ is $\ctw = 0$, since we can compose the $k$ RPQs into a single RPQ $r$ and rewrite $Q_1$ into the free-connex query $Q_1'(X_1,X_{k+1}) = r(X_1,X_{k+1})$. The free-connex fractional hypertree width of $Q_1$ is $\fw = 2$. Thus, $Q_1$ can be computed in time $\bigO{|E| + |E|\cdot \out^{1/2} + \out}$ using our algorithm (Theorem~\ref{theo:general}) and in time $\bigO{|E| + |E|\cdot\out_a^{1/2}+\out_a\cdot\out^{1/2}+\out}$ using \ospgosyan. Depending on the scale of $\out_a$, one algorithm is asymptotically faster than the other. Our algorithm outperforms \ospgosyan when $\out_a = \omega(|E|)$ or $\out_a = \omega(\out)$, i.e., when the RPQ outputs are asymptotically larger than $|E|$ or the query output size $\out$. 
    \nop{The key reason for our improvement is that the chain of $k$ RPQs can be contracted into a single RPQ without incurring an increase in the data complexity for its evaluation (the query complexity increases by a factor $k$). The query obtained by contraction is one RPQ and therefore trivially free-connex and can be computed using \ospg.
    }
    \qed
\nop{
    Let the $k$-path CRPQ $Q_1(X_1,X_{k+1}) = r_1(X_1,X_2),\ldots,r_{k}(X_{k},X_{k+1})$, where $r_1,\ldots,r_k$ are RPQs. The contraction width of $Q_1$ is $\ctw = 0$ and the free-connex fractional hypertree width is $\fw = 2$. Thus, $Q_1$ can be computed in time $\bigO{|E| + |E|\cdot \out_c^{1/2} + \out}$ using our algorithm (Proposition~\ref{prop:free-connex}) and in time $\bigO{|E| + |E|\cdot\out_a^{1/2}+\out_a\cdot\out^{1/2}+\out}$ using \ospgosyan. Since $\out_c \leq \out_a$, the complexity of our algorithm has at least the  additive factor $\out_a\cdot\out^{1/2}$ less than the best prior approach. This extra factor can be made arbitrarily large and can dominate all other factors in the complexity. The key reason for our improvement is that the chain of $k$ RPQs can be contracted into a single RPQ without incurring an increasing in the data complexity for its evaluation (the query complexity increases by a factor $k$). The query obtained by contraction is one RPQ and therefore trivially free-connex and can be computed using \ospg.\qed
}
\end{example}

\begin{example}
\rm
    Let the $k$-star CRPQ $Q_2(X_1,\ldots,X_{k}) = r_1(Y,X_1),\ldots,r_{k}(Y,X_{k})$, where $r_1,\ldots,r_k$ are RPQs for $k>2$. The contraction width of $Q_2$ is $\ctw=1$ and the free-connex fractional hypertree width is $\fw = k$. Thus, $Q_2$ can be computed in time $\bigO{|E| + |E|\cdot\out^{1/2} \cdot |V|^{1/2} + \out\cdot |V|}$ using our algorithm (Theorem~\ref{theo:general}) and in time $\bigO{|E| + |E|\cdot\out_a^{1/2}+\out_a\cdot\out^{1-1/k}+\out}$ using \ospgosyan. A fine-grained comparison of the two approaches is rather daunting. Yet it can be observed that depending on the scale of $\out_a$ and $\out$, one algorithm incurs a lower runtime than the other. One regime where our algorithm is faster is for large RPQ outputs but small overall query output: for instance, for $\out_a=\bigO{|V|^2}$ and $\out=o(1)$, our algorithm is faster by an $\bigO{|V|^{1/2}}$ factor. Our algorithm achieves the stated complexity as follows. It first promotes the bound variable $Y$ to free, so that the obtained query $Q_2'$ becomes free-connex. Then it evaluates $Q_2'$ with the same complexity as computing RPQs using \ospg, i.e., in time $\bigO{|E|+|E|\cdot\out^{1/2}}$. The output of $Q_2'$ can be a multiplicative factor $|V|$ larger than the output of $Q_2$. The complexity thus has $\out\cdot|V|$ in place of $\out$. \qed  
\end{example}

The rest of the paper is organized as follows.
Section~\ref{sec:prelims} introduces basic concepts.
Section~\ref{sec:contraction} defines and illustrates the contraction  width for CRPQs.
Sections~\ref{sec:freeconnex} and~\ref{sec:general} detail our algorithms for free-connex and general acyclic CRPQs, respectively.
We conclude in Section~\ref{sec:conclusion}. 
Proof details appear in the appendix.
\section{Preliminaries}
\label{sec:prelims}
For a natural number $n$, we define $[n] = \{1,2,\ldots,n\}$. In case $n=0$, we have $[n]=\emptyset$.
\paragraph*{Regular Languages}
An {\em alphabet} $\Sigma$ is a finite set of symbols. 
The set of all strings over an alphabet $\Sigma$ is denoted as $\Sigma^*$.
A {\em language} over $\Sigma$ is a subset of $\Sigma^*$.
We use the standard definition of {\em regular expressions}
composed of symbols from an alphabet $\Sigma$, the empty string
symbol $\epsilon$, and the concatenation, union, and Kleene star operators. We denote by $L(r)$ the language defined by a regular expression $r$. 
A language is {\em regular} if it can be defined by a regular expression.

The {\em reverse} $\rev(w)$ of a string $w \in \Sigma^*$ is defined recursively as follows: 
$\rev(w) = \epsilon$ if $w = \epsilon$ and 
$\rev(w) = \sigma \cdot \rev(v)$ if $w = v\sigma$ for some $v \in \Sigma^*$
and $\sigma \in \Sigma$.
The reverse of a language $L$ is $L^R = \{\rev(w) \mid w \in L\}$. For every regular expression $r$, we can construct in $\bigO{|r|}$ time a regular expression $r^R$ that defines the reverse of $L(r)$~\cite{HopcroftMU07}.

\paragraph*{Edge-Labeled Graphs}
An  {\em edge-labeled graph} is a directed graph 
$G=(V, E, \Sigma)$, where $V$ is a set of vertices, $\Sigma$ is an alphabet, and 
$E \subseteq V \times \Sigma \times V$ is a set of labeled edges. 
A triple $(u,\sigma, v) \in E$ denotes an edge from vertex $u$ to vertex $v$ labeled by $\sigma$\footnote{
Without loss of generality, we consider edge-labeled graphs where every vertex has at least one incident edge. A single pass over the edge relation suffices to extract these vertices.}. 
A {\em path} $p$ in $G$ from vertex $u$ to vertex $v$ consists of a sequence $u = w_0, w_1 \ldots , w_k = v$ of vertices from $V$ and a sequence $\sigma_1, \ldots ,\sigma_k$ of symbols from $\Sigma$
for some $k \geq 0$ such that for
every $i \in [k]$, there is an edge $(w_{i-1}, \sigma_i, w_i) \in E$. We say that the string $\sigma_1\sigma_2 \cdots \sigma_k$
is the {\em label} of the path $p$.

The {\em symmetric closure} $\hat{G}$ of $G$ contains for each edge in $G$ labeled by some symbol $\sigma$, an additional edge in opposite direction labeled by a fresh symbol $\hat{\sigma}$. That is, 
$\hat{G} = (V, \hat{E}, \hat{\Sigma})$, where 
$\hat{\Sigma} = \Sigma \cup \{\hat{\sigma} \mid \sigma \in \Sigma\}$ and $\hat{E} = E \cup \{(v,\hat{\sigma},u) \mid (u,\sigma, v) \in E\}$. 

\paragraph*{Conjunctive Regular Path Queries}
A {\em regular path query (RPQ)} over an alphabet $\Sigma$
is a regular expression using $\Sigma$.
The output of an RPQ $r$ evaluated over an edge-labeled graph 
$G=(V, E, \Sigma)$ is the set of all pairs $(u,v)$ of 
vertices such that $G$ contains a path from $u$ to $v$
labeled by a string from $L(r)$.

A {\em conjunctive regular path query (CRPQ)}, or query for short, is of the form 
\begin{align}
Q(\bm F) = r_1(X_1,Y_1) \wedge \cdots \wedge r_n(X_n,Y_n),
\label{eq:CRPQ}
\end{align}
where: each $r_i$ is an RPQ; each $r_i(X_i,Y_i)$ is an 
{\em atom};
$\vars(Q) = \{X_i, Y_i \mid i \in [n]\}$ is the set of variables;
$\free(Q) = \bm F \subseteq \vars(Q)$ is the set of {\em free} variables;
$\bound(Q) = \vars(Q)\setminus\free(Q)$ is the set of {\em bound} variables; 
$\atoms(Q)$ is the set of atoms; and the conjunction of the atoms is the {\em body} of the query.
The query is called {\em full} if all its variables are free.

The {\em query graph}  of $Q$ is an undirected multigraph 
$G_Q = (V, E)$, where $V$ is the variable set of $Q$
and $E$ is a multiset that contains an undirected edge 
$\{X_i,Y_i\}$ for each atom $r_i(X_i,Y_i)$, i.e., 
$E = \{\hspace{-0.2em}\{\, \{X_i, Y_i\} \mid r_i(X_i, Y_i) \in \atoms(Q)\, \}\hspace{-0.2em}\}$.
Hence, an atom $r_i(X_i, Y_i)$ with $X_i = Y_i$ is represented by the singleton set $\{X_i\}$ in the query graph. The leftmost graph in Figure~\ref{fig:cw-example} visualizes the query graph (free variables are underlined) of the CRPQ in Example~\ref{ex:elimination_procedure}. 

A {\em cycle} in a query graph is 
a sequence $u_1, e_1, u_2, e_2, \dots ,e_{k-1},u_k$ of vertices 
$u_1, \dots ,u_k$ and edges $e_1, \dots ,e_{k-1}$
for $k \geq 2$
such that 
(1) $u_i \neq u_j$ for $i,j \in [k-1]$ and $i\neq j$, 
(2) $e_\ell \neq e_m$ for $\ell,m \in [k-1]$ and $\ell\neq m$, and (3) $u_1 = u_k$.
A query is called {\em acyclic} if its query graph is cycle-free, which means that it  is a forest. 
In particular, the graph of an acyclic query contains neither self-loops (i.e., edges of the form $\{u\}$) nor multiple edges between the same vertex pair (i.e., subquery of the form $r_1(X,Y)\wedge r_2(X,Y)$).
For any acyclic query, we can construct a {\em join tree}, which is a tree such that:  (1) the nodes of the tree are the atoms of the query, and (2) for each variable, the following holds: if the variable appears in two atoms, then it appears in all atoms on the unique path between these two atoms in the tree~\cite{Brault-Baron16}.
An acyclic query is called {\em free-connex} if in each  tree
of the query graph, the free variables occurring in the tree form a connected subtree~\cite{BraultPhD13}.
Any free-connex acyclic query has a {\em free-top} join tree, which is  a join tree where all atoms containing a free variable form a connected subtree including the root.

We define the semantics of the CRPQ $Q$ from  
Eq.~(\ref{eq:CRPQ}) over an edge-labeled graph $G = (V, E, \Sigma)$ using homomorphisms, which are mappings
from the query variables to the vertices of $G$~\cite{AnglesABHRV17}. 
A mapping $\mu: \free(Q) \rightarrow V$ is in the output of $Q$, if it can be extended to a mapping
$\mu': \free(Q) \cup \bound(Q) \rightarrow V$
such that $G$ has a path from $\mu'(X_i)$
to $\mu'(Y_i)$ labeled by a string from $L(r_i)$, 
for every $i \in [n]$.
For convenience, we represent a mapping $\mu$ by the tuple $(\mu(X))_{X \in \free(Q)}$, assuming a fixed order on the free variables.

The {\em calibrated} output of an RPQ $r_i$ in $Q$ consists of all tuples
in the projection of the output of $Q$
onto the variables $X_i$ and $Y_i$. Hence, the calibrated output of RPQs whose variables are both bound is empty.

Our evaluation algorithms use at intermediate steps  
{\em conjunctive queries} that have atoms of the form $R(X_1, \ldots, X_k)$, where the arity $k \geq 0$ of an atom is arbitrary and $R$ references a  $k$-ary relation (see Appendix~\ref{app:prelims}).
They also use queries of the form 
$Q(\bm F) = r(X,Y) \wedge \bigwedge_{i \in [k]} R_i(\bm X_i)$
where $r$ is an  RPQ and each $R_i$ is a unary or nullary materialized  relation over the vertex set of the input graph. 
The semantics of such queries is straightforward and given in 
Appendix~\ref{app:prelims}.

Our complexity measures and computational model are given in Appendix~\ref{app:prelims}. 

\nop{
\paragraph*{Complexity Measures}
\nop{The {\em evaluation problem} for a CRPQ $Q$ has the following input and output:

\smallskip
    \begin{tabular}{ll} 
        \textbf{Input} & An edge-labeled graph $G = (V,E, \Sigma)$\\
        \textbf{Output} & The set of all tuples in the output of $Q$ over $G$
    \end{tabular}

\smallskip 
}
We use data complexity, i.e., we consider the query fixed and of constant size and measure the time complexity of solving the evaluation problem for CRPQs in terms of the number $|V|$ of vertices and the number $|E|$ of  edges of the input graph, the  size $\out$ of the output of the query, and the maximum size $\out_c$ of the calibrated output of any RPQ in the query.

We consider the RAM model of computation and assume that each materialized relation over the vertex set of the input graph is implemented by a data structure that allows to look up a tuple in
the relation in constant time and enumerate all tuples with constant delay.
}

\section{A Width Measure for Acyclic CRPQs}
\label{sec:contraction}

In this section, we introduce the {\em contraction width}, a new measure to express the complexity of evaluating acyclic CRPQs. We do so in two stages. 
First, we introduce an elimination procedure for query graphs  that allows to remove or contract edges through the elimination of bound variables.
To obtain the contraction width of a CRPQ $Q$, we take a free-connex tree decomposition of $Q$ and 
apply the elimination  procedure to each query induced by a bag of the decomposition. The contraction width of the decomposition is the maximal number of remaining bound variables in any  induced query. The contraction width of $Q$ is the minimal contraction width of any free-connex tree decomposition of $Q$.
\nop{
\dans{This definition of contraction width differs from the one given in Sec.~\ref{sec:intro}}
}

\subsection{Elimination Procedure for CRPQ Graphs}
\label{sec:elimination_procedure}
Consider a cycle-free undirected graph $G = (V, E)$.
The {\em degree} of a vertex $u$ in $G$ is the number of its neighbors, i.e., the number of vertices $v \in V$ with $\{u,v\} \in E$. We define two contraction operations on $G$ that eliminate vertices of degree 1 or 2: 
\begin{itemize}
\item {\em Elimination of a degree-2 vertex}: 
A graph $G'$ is the result of the elimination of a degree-2
vertex $u$ in $G$ if
$\{u,v_1\}, \{u,v_2\} \in E$ for distinct $v_1, v_2 \in V$
and $G'= (V \setminus \{u\}, E \setminus \{\{u,v_1\}, \{u,v_2\}\} \cup \{\{v_1, v_2\}\})$. That is, $G'$ is obtained from $G$ by removing $u$ and the two edges to its only neighbors $v_1$ and $v_2$ and adding a new edge connecting $v_1$ and $v_2$.

\item {\em Elimination of a degree-1 vertex}:
A graph $G'$ is the result of the elimination of a degree-1 vertex $u$ in $G$ if $\{u,v\} \in E$ for some $v \in V$ and $G' = (V', E')$, such that: (1) $E' = E \setminus \{\{u,v\}\}$ and; (2) $V' = V \setminus \{u, v\}$
if there is no $v'$ with $(v,v') \in E'$,  or $V' = V \setminus \{u\}$ otherwise.
In other words, $G'$ results from $G$ by removing the vertex $u$, the edge to its only neighbor $v$, and also $v$ if $v$ has no neighbor other than $u$.

\nop{
\dans{This does require a short explanation.  If you simply remove the regular expression $R(v,v')$ then the new query is no longer equivalent.  Also, a minor comment.  You don't need to consider two cases above: if $Q$ is connected, then $v$ is guaranteed to have an edge to some $v'$.}
}

\nop{
\item {\em Elimination of a degree-0 vertex}: 
A graph $G'$ is the result of the elimination of a degree-0 vertex $u$ in $G$ if $G' = (V \setminus \{u\}, E)$, i.e.,   $G'$ results from $G$ by removing vertex $u$.
}
\end{itemize}
We refer to the graph $G'$ obtained from $G$ after eliminating a vertex $u$ by $\eliminate(G,u)$.

Consider an acyclic  CRPQ $Q$ and its query graph $G_Q$.
An {\em elimination order} for $Q$ is a sequence 
$\omega= (X_1, \ldots, X_k)$ of bound variables such that 
there is a sequence of undirected graphs $(G_0, G_1, \ldots , G_k)$
with $G_0 = G_Q$, $G_i = \eliminate(G_{i-1}, X_i)$ for $i \in [k]$, and all bound variables in $G_k$ have degree at least 3.
We define $\omega(G_Q) = (G_1, \ldots , G_k)$ and
call $G_k$ a {\em contracted  query graph} of $Q$.
Any elimination order leads to the same contracted query graph:

\begin{restatable}{proposition}{UniqueReducedGraph}
\label{lem:unique-reduced-graph}
Every acyclic  CRPQ has a unique contracted query graph.
\end{restatable}

We denote the contracted query graph of a CRPQ $Q$ by $\cg(Q)$.

\begin{example}
\label{ex:elimination_procedure}
\rm Consider the following CRPQ whose query graph is depicted on the left in Figure~\ref{fig:cw-example} with free variables being underlined:
\nop{
\begin{align*}
Q(A, C, E, G, H) = r_1(A, B) \wedge &r_2(B, C) \wedge r_3(B, D) \wedge r_4(E,D) \wedge \,\\
& r_5(A, F) \wedge r_6(F, G) \wedge r_7(F, H) \wedge r_8(F, I) \wedge r_9(I, J)
\end{align*}
}
$Q(A, C, E, G, H) = r_1(A, B) \wedge r_2(B, C) \wedge r_3(B, D) \wedge r_4(E,D) \wedge
r_5(A, F) \wedge r_6(F, G) \wedge r_7(F, H) \wedge r_8(F, I) \wedge r_9(I, J)$.
The degree of the bound variable $D$ is 2, so we eliminate it and obtain the second graph in Figure~\ref{fig:cw-example}.  
Next, we eliminate the bound variable $J$ with degree 1, which results in the third graph in Figure~\ref{fig:cw-example}. 
Finally, we eliminate the bound variable $I$ with degree 1, resulting in the fourth  graph in Figure~\ref{fig:cw-example}. 
This graph is the contracted query graph of $Q$, since the only remaining bound variables $B$ and $F$ have both degree 3. 
\qed
\end{example}
\begin{figure}[t]
    \centering
     	\begin{tikzpicture}
      \tikzset{
        free/.style={minimum size=0.3cm, inner sep=0pt, text=black},
        bound/.style={minimum size=0.3cm, inner sep=0pt, text=black},
      }

      \node[free] (a) {\underline{A}};
      \node[bound, below left=0.3cm of a] (b) {B};
      \node[free, below left = 0.3cm of b] (c) {\underline{C}};
    \node[bound, below = 0.3cm of b] (d) {D};
    \node[free, below = 0.3cm of d] (e) {\underline{E}};
      \node[bound, below right=0.3cm of a] (f) {F};
      \node[free, below left = 0.3cm of f] (g) {\underline{G}};
    \node[free, below = 0.3cm of f] (h) {\underline{H}};
    \node[bound, below right = 0.3cm of f] (i) {I};
    \node[bound, below = 0.3cm of i] (j) {J};    

      \draw (a) -- (b);
      \draw (b) -- (c);
      \draw (b) -- (d);	
        \draw (d) -- (e);
        \draw (a) -- (f);
        \draw (f) -- (g);
        \draw (f) -- (h);
        \draw (f) -- (i);
        \draw (i) -- (j);     

      \begin{scope}[xshift=3.5cm]

      \node[free] (a) {\underline{A}};
      \node[bound, below left=0.3cm of a] (b) {B};
      \node[free, below left = 0.3cm of b] (c) {\underline{C}};
    \node[free, below = 0.5cm of b] (e) {\underline{E}};
      \node[bound, below right=0.3cm of a] (f) {F};
      \node[free, below left = 0.3cm of f] (g) {\underline{G}};
    \node[free, below = 0.5cm of f] (h) {\underline{H}};
    \node[bound, below right = 0.3cm of f] (i) {I};
\node[bound, below = 0.3cm of i] (j) {J};

      \draw (a) -- (b);
      \draw (b) -- (c);
      \draw (b) -- (e);	
        \draw (a) -- (f);
        \draw (f) -- (g);
        \draw (f) -- (h);
        \draw (f) -- (i);
        \draw (i) -- (j);

        \end{scope}
        
      \begin{scope}[xshift=7cm]

      \node[free] (a) {\underline{A}};
      \node[bound, below left=0.3cm of a] (b) {B};
      \node[free, below left = 0.3cm of b] (c) {\underline{C}};
    \node[free, below = 0.3cm of b] (e) {\underline{E}};
      \node[bound, below right=0.3cm of a] (f) {F};
      \node[free, below left = 0.3cm of f] (g) {\underline{G}};
    \node[free, below = 0.3cm of f] (h) {\underline{H}};
    \node[bound, below right = 0.3cm of f] (i) {I};

      \draw (a) -- (b);
      \draw (b) -- (c);
      \draw (b) -- (e);	
        \draw (a) -- (f);
        \draw (f) -- (g);
        \draw (f) -- (h);
        \draw (f) -- (i);
   
      \end{scope}

        \begin{scope}[xshift=10.5cm]

      \node[free] (a) {\underline{A}};
      \node[bound, below left=0.3cm of a] (b) {B};
      \node[free, below left = 0.3cm of b] (c) {\underline{C}};
    \node[free, below = 0.3cm of b] (e) {\underline{E}};
      \node[bound, below right=0.3cm of a] (f) {F};
      \node[free, below left = 0.3cm of f] (g) {\underline{G}};
    \node[free, below = 0.3cm of f] (h) {\underline{H}};

      \draw (a) -- (b);
      \draw (b) -- (c);
      \draw (b) -- (e);	
        \draw (a) -- (f);
        \draw (f) -- (g);
        \draw (f) -- (h);
    
      \end{scope}
      
      \end{tikzpicture}

    \caption{First graph: Query graph of the CRPQ from Example~\ref{ex:elimination_procedure}. The free variables are underlined. Second, third, and fourth graph: The graphs obtained after eliminating the bound variables $D$, $J$, and $I$, respectively.
    \nop{\mahmoud{Why does $D$ still appear in the third figure?}
    \ahmet{fixed}
    }}
    \label{fig:cw-example}
\end{figure}
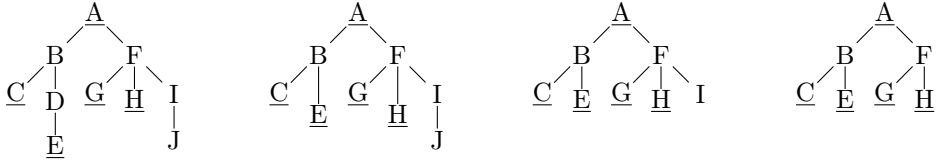
\subsection{Contraction Width}
We recall the definition of tree decompositions. 
%
    A tree decomposition of a CRPQ $Q$ is a pair 
    $\calT = (T, \chi)$, where $T$ is a tree with vertices $V(T)$  
    and $\chi: V(T) \rightarrow 2^{\vars(Q)}$ maps each vertex $t$ of $T$ to a subset $\chi(t)$ of variables of $Q$ such that the following properties hold:
    \begin{itemize}
        \item for every $r(X,Y) \in \atoms(Q)$, it holds $X, Y\in \chi(t)$ for some $t \in V(T)$, and
        \item for every $X \in \vars(Q)$, the set 
        $\{t \mid X \in \chi(t)\}$ is a non-empty connected subtree of $T$.
    \end{itemize}
    The sets $\chi(t)$ are called the {\em bags} of the tree decomposition. We denote the set of bags of $\calT$ by $\bags(\calT)$.
    The CRPQ $Q_{\calB}$ {\em induced} by a bag $\calB$ has the same body as $Q$ and has the free variables $\free(Q) \cap \calB$.\footnote{Our definition of induced query is tailored to CRPQs and differs from the standard notion of induced conjunctive queries. Given a tree decomposition for a conjunctive query $Q(\bm F) = \bigwedge_{i \in [n]} R_i(\bm X_i)$,
    the query induced by a bag $\calB$ is defined as $Q(\calB) = \bigwedge_{i \in [n]} R_i(\bm X_i \cap \calB)$. 
    Applying this definition to CRPQs can lead to induced queries that are not CRPQs, since their atoms can have less than two variables.}
    A tree decomposition $\calT$ is {\em free-connex} if $\calT$ contains a connected subtree such that the union of the bags in this subtree is the set $\free(Q)$ of free variables. 
    We use $\ftd(Q)$ to denote the set of all free-connex tree decompositions of the query $Q$. 

The contraction width of a tree decomposition is the maximal
number of bound variables in the contracted query graph of any induced query. 
The contraction width of a CRPQ is the minimal contraction width of any of  its free-connex tree decompositions:

\begin{definition}[Contraction Width]
\label{def:con-tree-width}
The contraction width $\ctw(Q)$ of any acyclic CRPQ $Q$ and
the contraction width $\ctw(\calT)$ of any tree decomposition $\calT$ of $Q$ are:
$$
\ctw(Q)  = \min_{\calT \in \ftd(Q)} \ctw(\calT)
\hspace{1cm} \text{and} \hspace{1cm} 
\ctw(\calT)  = \max_{\calB \in \bags(\calT)} |\bound(\cg(Q_\calB))|,
$$
where $\bound(\cg(Q_\calB))$ is the set of bound variables in
the contracted query graph of the query induced by the bag $\calB$.
\end{definition}

Note that two induced queries with different sets of free variables can impact differently the contraction width of the decomposition, even though they have the same body. This is because fewer free variables can enable the elimination of a larger number of bound variables.

The next example shows that a decomposition with several bags can lead to a lower contraction width compared to a trivial decomposition consisting of a single bag.

\begin{example}
\rm
Consider again the CRPQ $Q$ from Example~\ref{ex:elimination_procedure}
and the trivial tree decomposition for $Q$ with a single bag. The query induced by the single bag is the query $Q$ itself. 
Hence, the contraction width of this decomposition is 2, as 
shown in Example~\ref{ex:elimination_procedure}.

Now, consider the free-connex tree decomposition $\calT$ for $Q$ from Figure~\ref{fig:ctw-example} (left). The figure next shows the query graph of the induced query $Q_\calB$ for the bag $\calB = \{A,C,E\}$. The next graph is the query graph obtained after eliminating $D$. The second graph from the right is obtained after eliminating the variables $J$, $I$, $H$, and $G$ one after the other. 
We obtain the rightmost graph after eliminating $F$. This graph is the contracted query graph of $Q_\calB$, since its only bound variable $B$ is of degree 3. Hence, $|\bound(\cg(Q_\calB))| = 1$.  
Analogously, we can show that the number of bound variables 
in the contracted graphs of each of the other induced queries
is 1.
Hence, $\ctw(\calT) = 1$. It can be shown that $\ctw(Q) = 1$. 
\qed
\end{example}

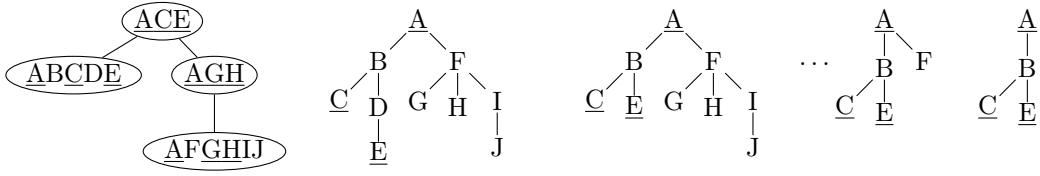
\begin{figure}[t]
    \centering
     	\begin{tikzpicture}
      \tikzset{
        free/.style={minimum size=0.3cm, inner sep=0pt, text=black},
        bound/.style={minimum size=0.3cm, inner sep=0pt, text=black},
      label/.style={minimum size=0.3cm, inner sep=0pt, text=black},
        bag/.style={ellipse, draw, minimum size=0.5cm, inner sep=0pt, text=black},
      }

      \node[bag] (free_left) {\underline{A}\underline{C}\underline{E}};
      \node[bag, below left=0.5cm of free_left, xshift=2mm] (all_left) {\underline{A}B\underline{C}D\underline{E}};
        \node[bag, below right=0.5cm of free_left, xshift=-4.5mm] (free_right) {\underline{A}\underline{G}\underline{H}};
      \node[bag, below=0.5cm of free_right] (all_right) {\underline{A}F\underline{G}\underline{H}IJ};

      \draw (free_left) -- (all_left);
      \draw (free_left) -- (free_right);
      \draw (free_right) -- (all_right);

\begin{scope}[xshift=3.4cm]
      \node[free] (a) {\underline{A}};
      \node[bound, below left=0.3cm of a] (b) {B};
      \node[free, below left = 0.3cm of b] (c) {\underline{C}};
    \node[bound, below = 0.3cm of b] (d) {D};
    \node[free, below = 0.3cm of d] (e) {\underline{E}};
      \node[bound, below right=0.3cm of a] (f) {F};
      \node[bound, below left = 0.3cm of f] (g) {G};
    \node[bound, below = 0.3cm of f] (h) {H};
    \node[bound, below right = 0.3cm of f] (i) {I};
   \node[bound, below  = 0.3cm of i] (j) {J};    

      \draw (a) -- (b);
      \draw (b) -- (c);
      \draw (b) -- (d);	
        \draw (d) -- (e);
        \draw (a) -- (f);
        \draw (f) -- (g);
        \draw (f) -- (h);
        \draw (f) -- (i);
        \draw (i) -- (j);        
 \end{scope}       
      
      \begin{scope}[xshift=6.8cm]

      \node[free] (a) {\underline{A}};
      \node[bound, below left=0.3cm of a] (b) {B};
      \node[free, below left = 0.3cm of b] (c) {\underline{C}};
    \node[free, below = 0.3cm of b] (e) {\underline{E}};
      \node[bound, below right=0.3cm of a] (f) {F};
      \node[bound, below left = 0.3cm of f] (g) {G};
    \node[bound, below = 0.3cm of f] (h) {H};
    \node[bound, below right = 0.3cm of f] (i) {I};
   \node[bound, below  = 0.3cm of i] (j) {J}; 
   
      \draw (a) -- (b);
      \draw (b) -- (c);
      \draw (b) -- (e);	
        \draw (a) -- (f);
        \draw (f) -- (g);
        \draw (f) -- (h);
        \draw (f) -- (i);
        \draw (i) -- (j);        
    
      \end{scope}

        \begin{scope}[xshift=8.7cm]

      \node (a) {$\,$};
      \node[below = 0.3cm of a] (b) {$\dots$};
    
      \end{scope}

      \begin{scope}[xshift=9.6cm]

      \node[free] (a) {\underline{A}};
      \node[bound, below=0.3cm of a] (b) {B};
      \node[free, below left = 0.3cm of b] (c) {\underline{C}};
    \node[free, below = 0.3cm of b] (e) {\underline{E}};
      \node[bound, below right =0.3cm of a] (f) {F};

      \draw (a) -- (b);
      \draw (b) -- (c);
      \draw (b) -- (e);	
        \draw (a) -- (f);
    
      \end{scope}

        \begin{scope}[xshift=11.5cm]

      \node[free] (a) {\underline{A}};
      \node[bound, below=0.3cm of a] (b) {B};
      \node[free, below left = 0.3cm of b] (c) {\underline{C}};
    \node[free, below = 0.3cm of b] (e) {\underline{E}};

      \draw (a) -- (b);
      \draw (b) -- (c);
      \draw (b) -- (e);	
    
      \end{scope}

      \end{tikzpicture}

    \caption{Left to right: Free-connex tree decomposition for the CRPQ $Q$ from Example~\ref{ex:elimination_procedure}; query graph of the query $Q_\calB$ induced by $\calB = \{A,C,E\}$; graph obtained after eliminating $D$; graph obtained after eliminating $J$, $I$, $H$, and $G$; contracted query graph obtained after eliminating $F$.}
    \label{fig:ctw-example}
\end{figure}

The next example illustrates that the  contraction width of free-connex acyclic queries is 0, as claimed in Proposition~\ref{prop:free-connex}:

\nop{
\begin{proposition}
\label{prop:ctw-free-connex}
    For any free-connex acyclic query $Q$, it holds $\ctw(Q) = 0$. 
\end{proposition}
}

\begin{example}
\label{ex:ctw-zero}
\rm
Consider the free-connex acyclic CRPQ
$Q(A, B, D) = r_1(A, B) \wedge r_2(B, C) \wedge r_3(B, D) \wedge r_4(A, E) \wedge r_5(E, F)$ and its free-connex tree decomposition $\calT$ consisting of the 
two bags $\calB_1 =\{A, B, D\}$ and $\calB_2 = \{A, B, C, E, F\}$. The induced query $Q_{\calB_1}$ is the same as $Q$. 
Figure~\ref{fig:ctw-zero} depicts on the left the query graph of $Q_{\calB_1}$ and right to it the graphs obtained after the elimination of the bound variables $F$, $C$, and $E$. The final contracted query graph has no bound variables. 
Similarly, we can show that contracted query graph of $Q_{\calB_2}$ has no bound variables.
This implies $\ctw(\calT) = \ctw(Q) = 0$.
\qed
\end{example}
\begin{figure}[t]
    \centering
     	\begin{tikzpicture}
      \tikzset{
        free/.style={minimum size=0.3cm, inner sep=0pt, text=black},
        bound/.style={minimum size=0.3cm, inner sep=0pt, text=black},
      label/.style={minimum size=0.3cm, inner sep=0pt, text=black},
        bag/.style={ellipse, draw, minimum size=0.5cm, inner sep=0pt, text=black},
      }

      \node[free] (a) {\underline{A}};
      \node[free, below =0.3cm of a] (b) {\underline{B}};
      \node[bound, below left = 0.3cm of b] (c) {C}; 
    \node[free, below = 0.3cm of b] (d) {\underline{D}};
    \node[bound, below right = 0.3cm of a] (e) {E};
      \node[bound, below = 0.3cm of e] (f) {F};

      \draw (a) -- (b);
      \draw (b) -- (c);
      \draw (b) -- (d);	
        \draw (a) -- (e);
        \draw (e) -- (f);
        
      \begin{scope}[xshift=3cm]

      \node[free] (a) {\underline{A}};
      \node[free, below =0.3cm of a] (b) {\underline{B}};
      \node[bound, below left = 0.3cm of b] (c) {C}; 
    \node[free, below = 0.3cm of b] (d) {\underline{D}};
    \node[bound, below right = 0.3cm of a] (e) {E};

      \draw (a) -- (b);
      \draw (b) -- (c);
      \draw (b) -- (d);	
        \draw (a) -- (e);
    
      \end{scope}

        \begin{scope}[xshift=6cm]

      \node[free] (a) {\underline{A}};
      \node[free, below =0.3cm of a] (b) {\underline{B}};
    \node[free, below = 0.3cm of b] (d) {\underline{D}};
    \node[bound, below right = 0.3cm of a] (e) {E};

      \draw (a) -- (b);
      \draw (b) -- (d);	
        \draw (a) -- (e);
    
      \end{scope}

      \begin{scope}[xshift=9cm]

      \node[free] (a) {\underline{A}};
      \node[free, below =0.3cm of a] (b) {\underline{B}};
    \node[free, below = 0.3cm of b] (d) {\underline{D}};

      \draw (a) -- (b);
      \draw (b) -- (d);	
    
      \end{scope}
      
      \end{tikzpicture}

    \caption{Left to right: Query graph of the query $Q_{\calB_1}$
    in Example~\ref{ex:ctw-zero}; graphs obtained after eliminating the variables $F$, $C$, and $E$.}
    \label{fig:ctw-zero}
\end{figure}
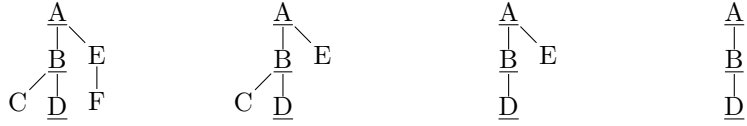

Another well-known width measure based on tree decompositions is the fractional hypertree width ($\fw$). 
We recall its definition in Appendix~\ref{sec:fract_hypertree_width}.  

\section{Evaluation of Free-Connex Acyclic CRPQs}
\label{sec:freeconnex}
In this section, we present our evaluation strategy for free-connex acyclic CRPQs and show that it runs within the time bound given in Proposition~\ref{prop:free-connex}. 
\nop{
$\bigO{|E| + |E|\sqrt{\out_c} + \out}$ time as claimed in Corollary~\ref{cor:free-connex}, where $E$ is the edge set of the input graph, $\out_c$ is the size of the calibrated outputs of the atoms, and $\out$ is the size of the output of the query. 
}

\subsection{Evaluation Strategy}
\label{sec:freeconnex-eval-strategy}
Given a free-connex acyclic CRPQ, we construct a free-top join tree for the query and perform, analogous to the Yannakakis algorithm~\cite{Yannakakis81}, a bottom-up and a top-down pass over the tree to compute the calibrated outputs of the RPQs occurring in the query.
We then join these output relations
to obtain the final output.
The key difference from the Yannakakis algorithm is that the outputs of the RPQs are not materialized prior to calibration.
We illustrate our strategy with an example.

\begin{figure}[t]
\hspace{2em}
        \begin{tikzpicture}
      \tikzset{
        bag/.style={ellipse, draw, minimum size=0.6cm, inner sep=0pt, text=black},
      }

    \node (xy) {$r_1(\underline{X},\underline{Y})$};
    \node at (-0.8,-1) (yz) {$r_2(\underline{Y},Z)$};
    \node at (0.8,-1) (xu) {$r_4(\underline{X},\underline{U})$};
    \node at (-0.8,-2) (zw) {$r_3(Z,W)$};

    \draw (xy) -- (yz);
    \draw (yz) -- (zw);
    \draw (xy) -- (xu);

\node (x) at (8.5,-1) {
\begin{minipage}{0.75\textwidth}
    $B_{3 \rightarrow 2}(Z) = r_3(Z, W)$  \hspace{1em}
    $B_{2 \rightarrow 1}(Y) = r_2(Y, Z) \wedge B_{3 \rightarrow 2}(Z)$ \\[0.2em]
$B_{4 \rightarrow 1}(X) = r_4(X, U)$ \hspace{1em}
 $R_1(X, Y)  = r_1(X, Y) \wedge B_{2 \rightarrow 1}(Y) \wedge B_{4 \rightarrow 1}(X)$     \\[0.2em]
 $T_{1 \rightarrow 2}(Y) = R_1(X, Y)$\hspace{1em} $T_{1 \rightarrow 4}(X) = R_1(X, Y)$ \\[0.2em]
$R_2(Y) = r_2(Y, Z) \wedge T_{1 \rightarrow 2}(Y)$ \hspace{1em} $R_4(X, U) = r_4(X, U) \wedge T_{1 \rightarrow 4}(X)$ \\[0.2em]
$Q'(X, Y, U) = R_1(X, Y) \wedge R_2(Y) \wedge R_4(X, U)$
\end{minipage}
};
    \end{tikzpicture}

    \caption{Left: Join tree for the free-connex query from 
    Example~\ref{ex:free-connex-eval}. Right: The materialized relations computed by our evaluation strategy.}
    \label{fig:fc-crpq-eval}
\end{figure}
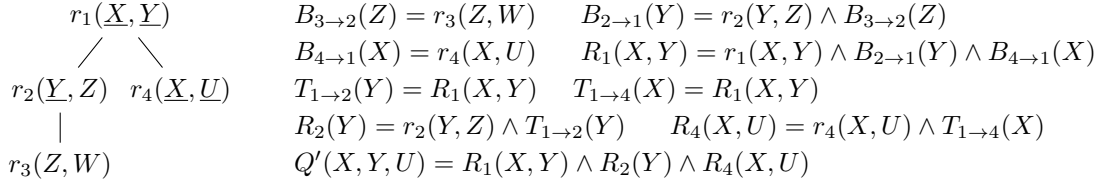

\begin{example}
\label{ex:free-connex-eval}
\rm
Consider the CRPQ $Q(X, Y, U) = r_1(X, Y) \wedge r_2(Y, Z) \wedge r_3(Z, W) \wedge r_4(X, U)$. 
Figure~\ref{fig:fc-crpq-eval} depicts a free-top join tree for the query, along with the materialized relations computed by our evaluation strategy.
The relations $B_{3\rightarrow 2}(Z)$, $B_{2\rightarrow 1}(Y)$, $B_{4\rightarrow 1}(X)$, and $R_{1}(X,Y)$ are computed during a bottom-up pass over the join tree, while the relations
$T_{1\rightarrow 2}(Y)$, $T_{1\rightarrow 4}(X)$, $R_{2}(Y)$, and $R_{4}(X,U)$ are materialized during a top-down pass. 
The relation $B_{3 \rightarrow 2}$ serves as a filter to discard tuples from the output of $r_2(Y,Z)$ that do not join with any tuple from the output of 
$r_3(Z,W)$. Similarly,  
$B_{2 \rightarrow 1}$ and $B_{4 \rightarrow 1}$ serve as filters to discard tuples from the output  of $r_1(X,Y)$ that do not appear in the final output.

The relations $T_{1\rightarrow 2}$ and $T_{1\rightarrow 4}$ filter out 
tuples from the outputs of $r_2(Y,Z)$ and $r_4(X,U)$, respectively,  that do not join with any tuple from the output of $R_1(X,Y)$.
The relation $R_1(X,Y)$ is the calibrated output of $r_1(X,Y)$, i.e., it is the projection of the output of $Q$ onto the variables $X$ and $Y$. 
The relations $R_2$ and $R_4$ are the calibrated outputs of $r_2(Y,Z)$ and $r_4(X,U)$, respectively. We obtain the final output of $Q$ by joining $R_1$, $R_2$, and $R_4$.
\qed
\end{example}

\begin{figure}[t]
\center
\begin{minipage}{\textwidth}
\center
\setlength{\tabcolsep}{3pt}
	\renewcommand{\arraystretch}{1}
	\renewcommand{\linenumber}{\makebox[2ex][r]{\rownumber\TAB}}
	\setcounter{magicrownumbers}{0}
	\begin{tabular}[t]{@{}l@{}}
		\toprule
 \textsc{EvalFreeConnex} (free-connex acyclic CRPQ 
 $Q(\bm F)$, edge-labeled graph $G$)\\
\midrule
\linenumber\hspace{-1em}
\LET $Q(\bm F) = r_1(X_1,Y_1) \wedge \cdots \wedge r_n(X_n,Y_n)$ \\
\linenumber\hspace{-1em} 
\LET $J$ be a free-top join tree for $Q$; 
\hspace{0.5em} \LET $r_i(X_i, Y_i)$ = root of $J$; \hspace{0.5em} \LET $\calR = \emptyset$ \\
\linenumber\hspace{-1em}   \textsc{BottomUp}$(r_i(X_i, Y_i), \bm F, J, \calR, G)$;  \hspace{1em} \textsc{TopDown}$(r_i(X_i, Y_i), \bm F, J, \calR, $G$)$; \\
\linenumber\hspace{-1em} \textbf{assume} $\calR$ is $\{R_{i_1}(\bm X_{i_1}), \dots, R_{i_k}(\bm X_{i_k})\}$   \\
\linenumber\hspace{-1em} \RETURN the output of 
$Q'(\bm F) = R_{i_1}(\bm X_{i_1})  \wedge \dots \wedge R_{i_k}(\bm X_{i_k})$\\
		\bottomrule
\end{tabular}
\end{minipage}

\vspace{0.5em}
\begin{minipage}{0.4\textwidth}
\setlength{\tabcolsep}{3pt}
	\renewcommand{\arraystretch}{1}
	\renewcommand{\linenumber}{\makebox[2ex][r]{\rownumber\TAB}}
	\setcounter{magicrownumbers}{0}
	\begin{tabular}[t]{@{}l@{}}
		\toprule
        \textsc{BottomUp}(atom $r_i(X_i, Y_i)$, variable set $\bm F$,\\
        \TAB tree $J$,  relation set $\calR$, edge-labeled graph $G$) \\
		\midrule
		\linenumber\hspace{-1em} \FOREACH $r_c(X_c,Y_c) \in \text{children}(r_i(X_i,Y_i), J)$  \\
		\linenumber\hspace{-1em} \TAB \textsc{BottomUp}$(r_c(X_c, Y_c), \bm F, J, \calR, G)$ \\
		 \linenumber\hspace{-1em}  \IF $r_i$ has no parent \\
        \linenumber\hspace{-1em} \TAB\LET $R_i(\{X_i, Y_i\} \cap \bm F) = r_i(X_i, Y_i)\wedge$ \\
        \TAB\ \ \ $\bigwedge\limits_{r_c(X_c, Y_c)\in\text{children}(r_i(X_i,Y_i), J)} \hspace{-4.5em}B_{c\rightarrow i} (\{X_c, Y_c\} \cap \{X_i, Y_i\})$ \\
    \linenumber\hspace{-1em} \TAB\textbf{update} $\calR$ to $\calR \cup \{R_i\}$ \\    
    \linenumber\hspace{-1em} \ELSE \\
        \linenumber\hspace{-1em} \TAB \LET $r_p(X_p, Y_p) = \text{parent}(r_i(X_i,Y_i), J)$ \\
        \linenumber\hspace{-1em}  \TAB\LET $B_{i\rightarrow p}(\{X_i, Y_i\}\cap \{ X_p, Y_p\}) = r_i(X_i, Y_i)\wedge$ \\
        \TAB $\bigwedge\limits_{r_c(X_c, Y_c)\in\text{ children}(r_i(X_i,Y_i), J)} \hspace{-4em}B_{c\rightarrow i} (\{X_c, Y_c\} \cap \{X_i, Y_i\})$\\
		\bottomrule
\end{tabular}
\end{minipage}
\hspace{4.9em}
\begin{minipage}{0.4\textwidth}
	\setlength{\tabcolsep}{3pt}
	\renewcommand{\arraystretch}{1.1}
	\renewcommand{\linenumber}{\makebox[2ex][r]{\rownumber\TAB}}
	\setcounter{magicrownumbers}{0}
	\begin{tabular}[t]{@{}l@{}}
		\toprule
        \textsc{TopDown}(atom $r_i(X_i, Y_i)$, variable set $\bm F$,\\
        tree $J$,  relation set $\calR$, edge-labeled graph $G$) \\
		\midrule 
    \linenumber\hspace{-1em} \IF $\{X_i, Y_i\} \cap \bm F =\emptyset$ \\
    \linenumber\hspace{-1em} \TAB \RETURN \\
    \linenumber\hspace{-1em} \IF $r_i$ has parent \\ 
     \linenumber\hspace{-1em} \TAB  \LET $r_p(X_p, Y_p) = \text{parent}(r_i(X_i, Y_i), J)$ \\
      \linenumber\hspace{-1em} \TAB\LET $T_{p\rightarrow i}(\{X_i, Y_i\}\cap \bm X) = R_p(\bm X)$\\
    \linenumber\hspace{-1em} \TAB\LET $R_i(\{X_i, Y_i\} \cap \bm F) = r_i(X_i, Y_i)\wedge$ \\
    \TAB\TAB\TAB\TAB\TAB\TAB\TAB\TAB\TAB\ \ $T_{p\rightarrow i}(\{X_i, Y_i\} \cap \bm X)$ \\
    \linenumber\hspace{-1em} \TAB\textbf{update} $\calR$ to $\calR \cup\{R_i\}$ \\
    \linenumber\hspace{-1em}  \FOREACH $r_c(X_c, Y_c)\in \text{children}(r_i(X_i,Y_i), J)$ \\
    \linenumber\hspace{-1em} \TAB \textsc{TopDown}($r_c(X_c, Y_c), \bm F, J, \calR, G$)\\
		\bottomrule
	\end{tabular}
\end{minipage}
\caption{Evaluating a free-connex acyclic CRPQ 
$Q(\bm F)$. 
We denote by $\text{parent}(r_i(X_i,Y_i), J)$ the parent and by 
$\text{children}(r_i(X_i,Y_i), J)$ the set of children of atom $r_i(X_i,Y_i)$ in the tree $J$. The relations of the form $R_i$, $B_{i \rightarrow j}$, and $T_{i \rightarrow j}$ are computed over the edge-labeled graph $G$.}
\label{fig:algo_free_connex}
\end{figure}

The procedure~\textsc{EvalFreeConnex} in Figure~\ref{fig:algo_free_connex} describes our evaluation strategy for a free-connex acyclic CRPQ 
$Q$ and an edge-labeled graph $G$.
It starts with 
constructing a free-top join tree $J$ for $Q$,
which means that the atoms containing a free variable form a connected subtree in $J$ including the root. 
The bottom-up and top-down passes are described by the subprocedures 
\textsc{BottomUp} and \textsc{TopDown}, respectively.
During these passes, the procedure computes, among other relations, a set $\mathcal{R}$ of relations $R_i$ whose join  is the output of $Q$.
This set consists of the calibrated output of the root atom and of all atoms with a free variable.
Next, we describe the procedures \textsc{BottomUp} and \textsc{TopDown} in more detail.

The main purpose of \textsc{BottomUp} is to materialize the calibrated output of the root atom in the join tree. It computes at each non-root atom $r_i(X_i,Y_i)$
a relation $B_{i \rightarrow p}$ that joins $r_i(X_i,Y_i)$
with the relations $B_{c \rightarrow i}$ computed at the children of $r_i(X_i,Y_i)$ and projects the result onto the common variables of 
$r_i(X_i,Y_i)$ and its parent atom  (Lines~7-8 of \textsc{BottomUp}).
If $r_i(X_i,Y_i)$ is the root, the procedure computes the relation $R_i$, defined analogously to $B_{i \rightarrow p}$ except that its set of variables is the intersection of $\{X_i, Y_i\}$ with the free variables of $Q$ (Lines 3–4 of \textsc{BottomUp}). The relation $R_i$ is the calibrated output of $r_i(X_i,Y_i)$.

The main purpose of the subprocedure \textsc{TopDown} is to materialize at each atom $r_i(X_i, Y_i)$ that is not the root of the join tree and contains at least one free variable of $Q$, the calibrated output 
$R_i$ of the RPQ $r_i$. For this, the procedure computes at each such atom $r_i(X_i, Y_i)$ with parent $r_p(X_p, Y_p)$,
a relation $T_{p\rightarrow i}$
that is the projection of the relation $R_p$, computed at the parent atom, onto the common variables of $R_p$ and $r_i(X_i, Y_i)$. Then, it materializes  the relation $R_i$ by computing the projection of the join of $r_i(X_i,Y_i)$ and  $T_{p\rightarrow i}$ onto  the free variables of $Q$ (Lines 4–6 of \textsc{TopDown}).

After the execution of \textsc{BottomUp} and \textsc{TopDown},
 the set $\calR = \{R_{i_1}, \ldots , R_{i_k}\}$ consists of the calibrated outputs of the root atom of the join tree and all other atoms that have at least one free variable. The procedure returns the output of the conjunctive query $Q'(\bm F) = \bigwedge_{j\in [k] }R_{i_j}$ (Line 5 of \textsc{EvalFreeConnex}). The output of this query is the output of $Q$ . We can show:

\begin{restatable}{proposition}{ComplexityEvalFreeConnex}
\label{prop:complexity_eval_free_connex}
    Given a free-connex acyclic CRPQ $Q$ and an edge-labeled graph 
    $G= (V, E, \Sigma)$, the procedure \textsc{EvalFreeConnex} 
    evaluates $Q$ in 
    $\bigO{|E| + |E|\cdot \out_c^{1/2} + \out}$ data complexity, 
    where $\out_c$ is the maximum size of the calibrated output of any RPQ in $Q$.
\end{restatable}

\subsection{Time Analysis}
We sketch the proof of the time complexity in Proposition~\ref{prop:complexity_eval_free_connex}.
The full proof is given in Appendix~\ref{app:freeconnex}.
Since the query size is considered fixed, a free-top join tree for $Q$ can be constructed in constant time. 
The relations $R_{i_1}, \ldots, R_{i_k}$ occurring  in the query $Q'$ in Line~4 of the procedure \textsc{EvalFreeConnex} are the calibrated outputs of RPQs that appear in a connected subtree of the join tree for $Q$. This means that the query $Q'$ forms a full acyclic conjunctive query.
Hence, its output can be computed in 
$\bigO{\max_{j \in k}|R_{i_j}| + \out}$ time using the Yannakakis algorithm~\cite{Yannakakis81}. Since the relations $R_{i_j}$
are calibrated, the time complexity simplifies to $\bigO{\out}$.
It remains to show that   
all relations constructed by the subprocedures \textsc{BottomUp} and \textsc{TopDown} can be computed in $\bigO{|E| + |E| \cdot \out_c^{1/2}}$ time.
To show this, we use the following lemma:
\begin{restatable}{lemma}{ReachableSetOSPGVariant}
\label{lem:reachable_set_ospg_variant}
Consider an edge-labeled graph $G = (V, E, \Sigma)$ and the
query 
$Q(\bm X) = r(X, Y) \wedge 
\bigwedge_{i \in [k]} U_i(X)$ $\wedge$
$\bigwedge_{i \in [\ell]} W_i(Y)$ $\wedge$
$\bigwedge_{i \in [m]} Z_i()$
for some $k, \ell, m \in \mathbb{N}$, where 
$r(X, Y)$ is an RPQ and each $U_i$, $W_i$, and $Z_i$
is a materialized relation. The query $Q$
can be evaluated in $\bigO{|E|}$ data complexity if $\bm F \subsetneq \{X,Y\}$, and 
in $\bigO{|E| + |E| \cdot \out^{1/2}}$ data complexity if $\bm F = \{X,Y\}$, where $\out$ is the output size of the query.
\end{restatable}

Since the query is acyclic, it cannot contain two atoms with the same set of variables. Hence, the definition of each relation 
$B_{i\rightarrow p}$ computed in the procedure \textsc{BottomUp} 
is of the form 
$V(\bm X) = r_i(X_i, Y_i) \wedge 
\bigwedge_{j \in [k]} U_j(X_i)$ $\wedge$
$\bigwedge_{j \in [\ell]} W_j(Y_i)$ $\wedge$
$\bigwedge_{j \in [m]} Z_j()$
for some relations $U_j$, $W_j$, and $Z_j$ and $k, \ell, m \in \mathbb{N}$, where 
$\bm X \subset \{X_i, Y_i\}$.
By Lemma~\ref{lem:reachable_set_ospg_variant}, 
the relation can be computed in time 
$\bigO{|E|}$.

The definition of the relations $R_i$ computed in the procedures \textsc{BottomUp} and \textsc{TopDown} is of the same form, except 
that the set of variables $\bm X$ is not necessarily a strict subset of $\{X_i, Y_i\}$. By Lemma~\ref{lem:reachable_set_ospg_variant}, each relation $R_i$ can be computed in $\bigO{|E| + |E| \cdot \out^{1/2}}$ time.
Since $R_i$ is the calibrated output of $r_i(X_i, Y_i)$, we can write the complexity as $\bigO{|E| + |E| \cdot \out_c^{1/2}}$.

Each relation $T_{p \rightarrow i}$ computed in \textsc{TopDown} is a projection of the relation $R_p$ and can thus be computed by a single pass over $R_p$. Since $R_p$ is computed in $\bigO{|E| + |E| \cdot \out_c^{1/2}}$ time, its size is upper-bounded by the same asymptotic complexity. Hence, each relation $T_{p \rightarrow i}$ can be computed in $\bigO{|E| + |E| \cdot \out_c^{1/2}}$ time.

\section{Evaluation of General Acyclic CRPQs}
\label{sec:general}
We present our evaluation strategy for general acyclic CRPQs, 
which follows free-connex tree decompositions. 
Given an acyclic CRPQ $Q$ and a free-connex tree decomposition for $Q$, we first create CRPQs induced by the bags of the decomposition. Then, we evaluate each induced query $Q'$ as follows. 
We begin by eliminating bound variables following an elimination order for $Q'$ (see Section~\ref{sec:elimination_procedure})
while preserving the query output.
Then, we  promote the remaining bound variables of $Q'$ to free variables, which turns $Q'$ into a free-connex acyclic query $Q''$.
Afterwards, we evaluate $Q''$ using our strategy for free-connex acyclic queries (see Section~\ref{sec:freeconnex}) and project the output onto the free variables of $Q'$.
Finally, we obtain the output of $Q$ by joining the outputs of all induced queries.

In Section~\ref{sec:elimination_bound_variables}, we explain how to remove bound variables in a CRPQ and adapt the input graph without changing the output of the query. In Section~\ref{sec:evaluation_through_decomposition}, we detail our overall evaluation strategy for acyclic CRPQs.

\subsection{Elimination of Bound Variables in CRPQs}
\label{sec:elimination_bound_variables}
Given a CRPQ and an input graph, we explain how to 
eliminate bound variables following an elimination order and adapt the input graph in linear time 
while preserving the output of the query. We illustrate our strategy by an example.

\begin{example}
\label{ex:elimination_bound_variables_CRPQ}
\rm
Consider again the CRPQ 
\nop{
\begin{align*}
Q(A, C, E, G, H) = r_1(A, B) \wedge & r_2(B, C) \wedge r_3(B, D) \wedge r_4(E,D) \wedge \,\\
& r_5(A, F) \wedge r_6(F, G) \wedge r_7(F, H) \wedge r_8(F, I), r_9(I, J)
\end{align*}
}
$Q(A, C, E, G, H) = r_1(A, B) \wedge  r_2(B, C) \wedge r_3(B, D) \wedge r_4(E,D) \wedge
 r_5(A, F) \wedge r_6(F, G) \wedge r_7(F, H) \wedge r_8(F, I), r_9(I, J)$
from Example~\ref{ex:elimination_procedure}.
The query graph of $Q$ and the graphs obtained after the elimination of each of the variables in the elimination order 
$(D, J, I)$ are depicted in Figure~\ref{fig:cw-example}. 
Let $G$ be the input graph.
We start with constructing the symmetric closure $\hat{G}$ of $G$. 
The graph $\hat{G}$ allows us to evaluate RPQs while traversing the edges of the original input graph in opposite direction.
Next, we explain how at each elimination step, the query $Q$ and the graph $\hat{G}$ are adjusted.

At the elimination step for the variable $D$,
we replace the atoms $r_3(B, D)$ and  $r_4(E,D)$ by the atom $r_3\hat{r}_4^{R}(B,E)$; here, $\hat{r}_4$ is obtained  from $r_4$ by replacing every symbol $\sigma$
with $\hat{\sigma}$ and 
$\hat{r}_4^R$ is the regular expression that defines the reverse language of $\hat{r}_4$. 
Observe that the input graph $G$ contains three vertices 
$u,v,w$ such that there is a path labeled by $r_3$ from $u$ to 
$v$ and a path labeled by $r_4$ from $w$ to $v$ if and only if 
$\hat{G}'$ has two vertices $u$ and $w$ connected by a path labeled by $r_3\hat{r}_4^{R}$.

At the elimination step for the variable $J$, we remove the atom $r_9(I,J)$ from the query. We also  compute a {\em filter relation} $R_I$ defined by $R_I(I) = r_9(I,J)$ that consists of all vertices in the input graph from which there is a path labeled by a string from $L(r_9)$. 

At the elimination step for the variable $I$,
we remove the atom $r_8(F,I)$ from the query. 
Then, we compute a filter relation $R_F$ defined by 
$R_F(F) = r_8(F,I) \wedge R_I(I)$
and add to each vertex $u \in R_F$ in the input graph, a self-loop edge $(u,\sigma_F,u)$ using a
fresh symbol $\sigma_F$. Finally, we add to the query the atom 
$\sigma_F(F,F')$ for some fresh variable $F'$.

Let the query and the input graph obtained after the above transformations be denoted as $Q'$ and $\hat{G}'$, respectively. 
The query $Q'$ is of the form 
$Q'(A, C, E, G, H) = r_1(A, B) \wedge r_2(B, C) \wedge r_3\hat{r}_4^R(B, E) \wedge r_5(A, F) \wedge r_6(F, G) \wedge r_7(F, H) \wedge \sigma_F(F, F')$. 
The contracted query graph of $Q$ (rightmost graph in Figure~\ref{fig:cw-example}) is the query graph of the query $Q'$ after skipping the atom $\sigma_F(F, F')$.
The output of $Q$ over $G$ is the same as the output of $Q'$ over $\hat{G}'$.
\qed
\end{example}

\begin{figure}[t]
\center
\setlength{\tabcolsep}{3pt}
	\renewcommand{\arraystretch}{1}
	\renewcommand{\linenumber}{\makebox[2ex][r]{\rownumber\TAB}}
	\setcounter{magicrownumbers}{0}
	\begin{tabular}[t]{@{}l@{}}
		\toprule
 \textsc{Contract} (acyclic CRPQ $Q$, edge-labeled graph $G = (V,E,\Sigma)$)\\
\midrule
\linenumber \LET $Q_0= Q$; \TAB\LET $G_0 = G_Q$\\ 
\linenumber
\LET $\omega = (X_1, \ldots , X_k)$ be an elimination order for $Q_0$ and $\omega(G_0) = (G_1, \ldots, G_k)$\\
\linenumber \FOREACH $X \in \vars(Q_0)$ \LET $R_X= V$ \hspace*{2em} \text{// } $R_X$ \text{ is a unary filter relation}\\
\linenumber \FOREACH $i \in [k]$\\
\linenumber \TAB  \IF $X_i$ is a degree-1 variable in $G_{i-1}$ \\
\linenumber \TAB\TAB \LET $r(\bm X) \in \atoms(Q_{i-1})$
such that $X_i$ and some variable $Y$ are contained in $\bm X$\\
\linenumber \TAB\TAB \LET $Q_i$ = the query $Q_{i-1}$ after removing the atom $r(\bm X)$\\
\linenumber \TAB \TAB  
\LET $R_Y'(Y) = R_Y(Y)$; \TAB\textbf{update} $R_Y(Y)$ \textbf{to} $R_Y(Y) = r(\bm X) \wedge R_Y'(Y) \wedge R_{X_i}(X_i)$ \\
\nop{
\linenumber \TAB \TAB update filter relation $R_Y$:  
$R_Y'(Y): = R_Y(Y)$ and  
 $R_Y(Y): = r(\bm X) \wedge R_Y'(Y) \wedge R_{X_i}(X_i)$ \\
}
\linenumber \TAB  \IF $X_i$ is a degree-2 variable in $G_{i-1}$ \\
\ \ \TAB\TAB wlog, assume that $r_1(Y,X_i), r_2(Z,X_i) \in \atoms(Q_{i-1})$ for distinct variables $Y$ and $Z$\\
\linenumber \TAB\TAB \LET $Q_{i}$ = the query $Q_{i-1}$ after replacing atoms $r_1(Y,X_i)$ and $r_2(Z,X_i)$
with $r_1\sigma_{X_i}\hat{r}_2^R(Y,Z)$ \\
\linenumber \LET $\hat{G} = (V, \hat{E}, \hat{\Sigma})$ be the symmetric closure of $G$ \\
\linenumber \LET $\hat{G}'  = (V,\hat{E}', \hat{\Sigma}')$ where $\hat{\Sigma}' = \hat{\Sigma} \cup \{\sigma_X \mid \sigma_X \notin \hat{\Sigma} \text{ and } X \in \vars(Q)\}$ \\
\TAB\TAB\TAB\TAB\TAB\TAB\TAB\TAB\TAB\TAB\TAB\TAB\ \ \ and $\hat{E}' = \hat{E} \cup \{(u,\sigma_X,u)\mid X \in \vars(Q) \text{ and } u \in R_X\}$\\
\linenumber \LET $Q_k'$ = the query $Q_k$ after adding the atom $\sigma_X(X,\bar{X})$ with new variable $\bar{X}$,  $\forall X \in \vars(Q_k)$  \\
\linenumber \RETURN $(Q_k',\hat{G}')$\\
\bottomrule
\end{tabular}
\caption{Eliminating bound variables in an acyclic CRPQ $Q$ following 
an elimination order $\omega$ as defined in Section~\ref{sec:elimination_procedure} and 
adjusting the input graph $G$.}
\label{fig:eliminate_bound_in_query}
\end{figure}

Given an acyclic CRPQ $Q$ and an input graph $G$, the procedure  \textsc{Contract}
in Figure~\ref{fig:eliminate_bound_in_query} removes bound variables in $Q$ following an elimination order $\omega = (X_1, \ldots , X_k)$  for the query 
and adjusts the input graph such that the query  output is preserved. 
The procedure starts with creating a collection of  unary filter relations $R_X$ with  $X \in \vars(Q)$ (Line~3). The purpose of such a  relation $R_X$ is to restrict the domain of the variable $X$. At the beginning, there are no restrictions on the domains, so each filter relation consists of the whole vertex set of the input graph. 
We define $Q_0:=Q$ and $G_0:=G_Q$, where $G_Q$ is the query graph of $Q$. Let
$\omega(G_0) = (G_1, \ldots , G_k)$ be the sequence of graphs obtained after each elimination step following the elimination order $\omega$. 
For each $X_i$ in $\omega$, the procedure eliminates the variable $X_i$ from the query $Q_{i-1}$ and obtains a query $Q_i$ as follows.

Consider the case that $X_i$ is a degree-1 variable in $G_{i-1}$
(Lines~5--8).
Assume that $r(\bm X)$ is the only atom in $Q_{i-1}$ such that $X_i$ and some other variable $Y$ are contained in $\bm X$. 
The procedure removes the atom from the query and updates the filter relation $R_Y$ by restricting the domain of $Y$ to all vertices in the input graph from which there is a path labeled by a string from $L(r)$ to a vertex in the domain of $X_i$. 

Consider now the case that $X_i$ is a degree-2 variable in $G_{i-1}$
(Lines~9--10).
The pseudocode in Figure~\ref{fig:eliminate_bound_in_query} treats, without loss of generality, the case that the two atoms containing $X_i$ are of the form $r_1(Y,X_i)$ and $r_2(Z,X_i)$.
The procedure replaces the two atoms with the single atom $r_1\sigma_{X_i}\hat{r}_2^R(Y,Z)$, where $\sigma_{X_i}$ is a fresh alphabet symbol.
As it will be clear when we extend the input graph in Line~12, 
the symbol $\sigma_{X_i}$ ensures that the path satisfying $r_1$
ends up at a vertex that is contained in the domain of variable $X_i$. The regular expression $\hat{r}_2^R$ makes sure that the graph contains a path from $Z$ to $X_i$ that satisfies $r_2$.
The two cases not treated in the pseudocode are handled completely analogously:  
in case the atom pair containing $X_i$ is of the form 
$(r_1(Y,X_i), r_2(X_i,Z))$ or $(r_1(X_i,Y), r_2(X_i,Z))$, it is replaced by the atom $r_1\sigma_{X_i}r_2(Y,Z)$ or the atom $\hat{r}_1^{R}\sigma_{X_i}r_2(Y,Z)$, respectively.

When no more bound variables can be eliminated, the procedure adjusts the input graph $G$ and the query $Q_k$ (obtained from $Q$ after eliminating the variables in $\omega$) as follows (Lines~11--13).  
It constructs the symmetric closure $\hat{G}$ of $G$ and, for each filter relation $R_X$, performs the following steps:
(1) It adds a self-loop edge labeled by the symbol $\sigma_X$ to every vertex in $\hat{G}$ that belongs to $R_X$; (2) it adds the atom 
$\sigma_X(X, \bar{X})$ to the query $Q_k$, where $\bar{X}$ is a fresh variable.
By construction, this ensures that every vertex in the input graph that is mapped to variable $X$ is contained in the filter relation $R_X$.
\nop{In principle, it suffices to add the self-loops labeled by $X$ and the atoms $\sigma_X(X, \bar{X})$ 
only for those variables $X$ for which $R_X$ does not include all vertices in $V$. To keep the procedure simple, we do not consider this optimization here.}

The initialization of the filter relations in Line~3 and their update in Line~8 can be done in time 
$\bigO{|E|}$ (by Lemma~\ref{lem:reachable_set_ospg_variant}).
The symmetric closure $\hat{G}$ can be constructed by a single pass over the edge relation of $G$. 
For each filter relation $R_X$, adding to the symmetric closure self-loops of the form 
$(u, \sigma_X, u)$ for each $u \in R_X$ takes also $\bigO{|V|} = \bigO{|E|}$ time. Hence:

\begin{restatable}{proposition}{ContractQuery}
\label{prop:contract_query}
    Given an acyclic CRPQ $Q$ and an edge-labeled graph 
    $G= (V, E, \Sigma)$, the procedure \textsc{Contract} 
    creates in $\bigO{|E|}$ data complexity a CRPQ 
    $Q'$ and an edge-labeled graph $G'$ such that the output of $Q$ over $G$ is the same     as the output of $Q'$ over $G'$.
\end{restatable}

\subsection{Evaluation via Decomposition}
\label{sec:evaluation_through_decomposition}
Our overall evaluation strategy for acyclic CRPQs decomposes a given CRPQ using a free-connex tree decomposition, computes the outputs of the queries induced by the bags and joins these outputs to obtain the output of the original CRPQ. We start with an example. 

\begin{example}
\label{ex:evaluation_general}
\rm
We continue with the CRPQ $Q$ from Example~\ref{ex:elimination_bound_variables_CRPQ}.
Figure~\ref{fig:ctw-example} gives on the left a free-connex tree decomposition for $Q$ with the bags
$\calB_1 = \{A,C,E\}$, $\calB_2 = \{A,B, C,D,E\}$, $\calB_3 = \{A,G,H\}$, and $\calB_4 = \{A,F,G,H,I\}$.
For each bag $\calB_i$, we have an induced query $Q_{\calB_i}$ whose body is the same as the body of 
$Q$ and whose set of free variables is the intersection of the bag with the free variables of $Q$.
For instance, the free variables of the
induced query $Q_{\calB_1}$ are $\{A,C,E\}$.
The procedure \textsc{Contract} in Figure~\ref{fig:eliminate_bound_in_query} converts 
$Q_{\calB_1}$ into a CRPQ $Q_{\calB_1}'$ by eliminating
all bound variables of
$Q_{\calB_1}$ up to variable $B$ and creating some fresh non-join variables.
We promote the bound variable $B$ to free and obtain the 
the query $Q_{\calB_1}''$ with free variables $\{A,B, C,E\}$. Since its bound variables are non-join variables, the query $Q_{\calB_1}''$ is free-connex. We use the procedure  \textsc{EvalFreeConnex} in Figure~\ref{fig:algo_free_connex} to evaluate $Q_{\calB_1}''$. After projecting the output of $Q_{\calB_1}''$ onto the variables $\{A,C,E\}$, we obtain the output of $Q_{\calB_1}$.
Similarly, we compute the outputs of the queries $Q_{\calB_2}$--$Q_{\calB_4}$.
Finally, we obtain the output of $Q$ by joining the outputs of 
$Q_{\calB_2}$--$Q_{\calB_4}$.
\qed
\end{example}

\begin{figure}[t]
\center
\setlength{\tabcolsep}{3pt}
	\renewcommand{\arraystretch}{1}
	\renewcommand{\linenumber}{\makebox[2ex][r]{\rownumber\TAB}}
	\setcounter{magicrownumbers}{0}
	\begin{tabular}[t]{@{}l@{}}
		\toprule
 \textsc{EvalAcyclic} (acyclic CRPQ $Q(\bm F)$, free-connex tree decomposition $\calT$, edge-labeled graph $G$)\\
\midrule
\linenumber \LET $Q_1(\bm F_1), \ldots, Q_k(\bm F_k)$ be the CRPQs
induced by the bags of $\calT$\\
\linenumber \FOREACH $i\in [k]$\\
\linenumber  \TAB \LET $(Q_i', G_i)= \textsc{Contract}(Q_i, G)$\\
\linenumber \TAB \LET $Q_i''$ = the query obtained by promoting the vars in $\bound(Q_i)\cap \bound(Q_i')$ to free in $Q_i'$  \\
\nop{\linenumber \TAB \LET $Q_i''$ = the result from $Q_i'$ by promoting the variables in $\bound(Q_i)\cap \bound(Q_i')$ to free \\}
\linenumber \TAB \LET $R_i(\bm X_i) = \textsc{EvalFreeConnex}(Q_i'', G_i)$\\
\linenumber \TAB \LET $R_i'(\bm F_i) = R_i(\bm X_i)$ \hspace*{4em} \text{// We project away} $\bm X_i\setminus\bm F_i$\\
\linenumber \RETURN the output of $Q'(\bm F) = R_1'(\bm F_1) \wedge \cdots \wedge R_k'(\bm F_k)$ \\
\bottomrule
\end{tabular}
\caption{Evaluating an acyclic CRPQ $Q$ over an edge-labeled graph $G$ following a free-connex tree decomposition $\calT$ for $Q$.
The procedures \textsc{Contract} and \textsc{EvalFreeConnex} are given in Figures~\ref{fig:eliminate_bound_in_query} and \ref{fig:algo_free_connex}, respectively.}
\label{fig:evaluate_acylcic}
\end{figure}

Given an acyclic CRPQ $Q$, a free-connex tree decomposition $\calT$ for $Q$, and an input graph $G$, the procedure \textsc{EvalAcyclic} in Figure~\ref{fig:evaluate_acylcic} evaluates $Q$ following the structure of the decomposition. The procedure first constructs the queries  
$Q_1, \ldots, Q_k$ induced by the bags of $\calT$. For each induced query 
$Q_i$, it performs the following steps. 
It executes the  procedure \textsc{Contract} from Figure~\ref{fig:eliminate_bound_in_query} on the pair $(Q_i, G)$ and obtains a transformed pair $(Q_i', G_i)$. It constructs 
the query $Q_i''$ by turning the bound variables in $Q_i'$
that are also contained in $Q_i$ into free variables. 
Then, it computes the output $R_i$ of the free-connex query $Q''_i$ on the edge-labeled graph $G_i$ using the procedure \textsc{EvalFreeConnex} from Figure~\ref{fig:algo_free_connex}.
Then, it obtains $R_i'$ from $R_i$ by projecting it onto the free variables of $Q_i$. Finally, it computes the output of $Q$ by joining the relations $R_1', \ldots, R_k'$.

\begin{restatable}{proposition}{ComplexityEvalAcyclic}
\label{prop:complexity_eval_acyclic}
    Given an acyclic CRPQ $Q$, a free-connex tree decomposition $\calT$ for $Q$, and an edge-labeled graph $G= (V, E, \Sigma)$, the procedure \textsc{EvalAcyclic}
    evaluates $Q$ in 
    $\bigO{|E| + |E|\cdot \out^{1/2} \cdot |V|^{\ctw/2} + \out \cdot |V|^{\ctw}}$ data complexity, 
    where $\out$ is the size of the output of $Q$ and $\ctw$ is the contraction width of $\calT$.
\end{restatable}

\begin{proof}
Consider an acyclic CRPQ $Q(\bm F)$, a free-connex tree decomposition 
$\calT$ for $Q$, and an edge-labeled graph $G= (V, E, \Sigma)$.
We first show that the procedure \textsc{EvalAcyclic} produces the correct output of $Q$ over $G$.  
Since the body of each induced query $Q_i$ is the same as the body of $Q$ and the union of the free variables of  the induced queries is the set $\bm F$, the join of the outputs of the induced queries is exactly the output of $Q$.
All bound variables in each query $Q_i''$ constructed in Line~4 are non-join variables, which implies that $Q_i''$ is free-connex.  Hence, \textsc{EvalFreeConnex} produces the correct output of  $Q_i''$ over $G$.
By construction, each $R_i'$ is the output of $Q_i$ over $G$.
Hence, the output of $Q'$ in Line~7 
is the output of $Q$.

Now, we show that the procedure \textsc{EvalAcyclic} runs within 
the time bound stated in the proposition.
For $i \in [k]$, we denote by $\out_i$, $\out_i'$, and $\out_i''$ the output sizes of the queries $Q_i$, $Q_i'$, and $Q_i''$, respectively.  
The execution of \textsc{Contract} on a pair $(Q_i, G)$ in Line~3 takes $\bigO{|E|}$ time (Proposition~\ref{prop:contract_query}).  
The procedure \textsc{EvalFreeConnex}
called in Line~5 runs in $\bigO{|E| + |E| \cdot 
\out''^{1/2}_c + \out''}$ time, where $\out_c''$ is the maximal size of the calibrated output of any RPQ occurring in $Q_i''$
(Proposition~\ref{prop:complexity_eval_free_connex}). 
Let $p$ be the number of variables in $Q_i'$ that are promoted from bound to free in Line~4. 
It holds 
$\out_c'' \leq \out_i''$ and $\out_i'' \leq \out_i' \cdot |V|^{p} = \out_i \cdot |V|^{p}$. Since the output of $Q_i$ is the projection of the output of $Q$ onto the free variables of $Q_i$, it also holds
$\out_i \cdot |V|^{p} \leq \out \cdot |V|^{p}$.
The contraction width $\ctw$ of $\calT$ is the maximal number of promoted variables 
over all bags of $\calT$. Hence, we can write the time complexity of \textsc{EvalFreeConnex} in Line~5 as
$\bigO{|E| + |E| \cdot \out^{1/2} \cdot |V|^{\ctw/2} + \out\cdot |V|^{\ctw}}$. The size of the output $R_i$ of 
\textsc{EvalFreeConnex} is bounded by the same complexity expression. The relation $R_i'$ in Line~6 can be computed by a single pass over the relation $R_i$, which takes time linear in the size of $R_i$. The query $Q'$ in Line~7 
constitutes a full acyclic conjunctive query. Hence, we can use the Yannakakis algorithm~\cite{Yannakakis81} to compute the output of $Q'$ in time
$\bigO{\max_{i \in [k]}|R_i'| + \out}$, which simplifies 
to $\bigO{\out}$ due to $\max_{i \in [k]}|R_i'| \leq \out$.
This implies that the overall time of the procedure \textsc{EvalAcyclic} is $\bigO{|E| + |E| \cdot \out^{1/2} \cdot |V|^{ \ctw/2} + \out\cdot |V|^{\ctw}}$.
\end{proof}

Using a free-connex tree decomposition $\calT$ for $Q$ with 
$\ctw(\calT) = \ctw(Q)$, Proposition~\ref{prop:complexity_eval_acyclic}
immediately implies Theorem~\ref{theo:general}. 
\section{Conclusion}
\label{sec:conclusion}

In this paper, we introduced an output-sensitive evaluation approach for acyclic conjunctive regular path queries. There are three immediate questions that naturally follow from this development. (1) Can our approach be further improved to match the complexity of the best-known approach for the output-sensitive evaluation of acyclic conjunctive queries~\cite{Hu25}, dubbed the output-sensitive Yannakakis in Section~\ref{sec:intro}? If this were true, then the complexity of acyclic conjunctive regular path queries would match the complexity of acyclic conjunctive queries, albeit the widely-held conjecture in the literature is that the former ought to be harder than the latter.
(2) Can our approach be extended to cyclic queries as well? It is already open whether the recent output-sensitive evaluation algorithm for acyclic conjunctive queries~\cite{Hu25} can be extended to cyclic conjunctive queries, let alone for cyclic conjunctive regular path queries. Yet there is hope: Earlier work introduced an output-sensitive algorithm for listing triangles~\cite{DBLP:conf/icalp/BjorklundPWZ14}. It is open whether the methodology for listing triangles in an output-sensitive manner can be extended to arbitrary cyclic queries. (3) Do output-sensitive algorithms yield lower runtimes than the classical non-output sensitive algorithms for practical workloads?

\bibliography{bibliography}

\appendix
\section{Missing Details in Section \ref{sec:intro}}
\label{app:intro}

\subsection{Proof of Theorem~\ref{theo:general}}
\MainTheorem*

\begin{proof}
Consider a CRPQ $Q$, an edge-labeled graph $G = (V, E, \Sigma)$  and a free-connex tree decomposition $\calT$ for $Q$ with $\cw = \ctw(Q) = \ctw(\calT)$. By Proposition~\ref{prop:complexity_eval_acyclic},  
$Q$ can be evaluated in 
$\bigO{|E| + |E| \cdot \out^{1/2} \cdot |V|^{\ctw/2} +\out \cdot |V|^{\ctw}}$ time. 
\end{proof}

\subsection{Proof of Proposition~\ref{prop:properties_contraction_width}}
\PropertiesContractionWidth*

\begin{proof}
Consider the infinite class of all k-star queries
of the form $Q_k(X_1, \ldots , X_{k}) = r_1(Y, X_1), \ldots , r_k(Y, X_k)$
for $k \geq 3$. The query graph of each query $Q_k$  contains a single bound variable of degree at least three. Hence, the contraction width of each query is 1. However, the fractional hypertree width of each query $Q_k$ is $k$.
\end{proof}

\subsection{Proof of Proposition~\ref{prop:free-connex}}
\FreeConnex*

\begin{proof}
    We prove the first statement in Proposition~\ref{prop:free-connex}.
Consider a free-connex acyclic CRPQ $Q$.
By definition, the query graph $G_Q$ of $Q$ consists of 
trees $T_1, \ldots , T_k$ such that in each tree, the free variables of the tree form a connected subtree.
In each tree $T_i$, we can repeatedly eliminate degree-1 vertices starting from the leaves of the tree. At the end of this process, each tree $T_i$ becomes empty or consists of only free variables. This implies that the contraction width of $Q$ is $0$.  

\smallskip
The second statement in Proposition~\ref{prop:free-connex}
follows immediately from Proposition~\ref{prop:complexity_eval_free_connex}.
\end{proof}

\subsection{Comparison with Related Work -- Examples}
\label{subsec:app_query_examples}

In the following, we compare the performance of our algorithm with the performance of \ospgosyan via illustrative examples.
The following table restates 
the evaluation times achieved by 
\ospgosyan and our algorithm for free-connex acyclic (first row) and general acyclic (second row) queries. We recall the parameters appearing in the complexity expressions:  
$V$ and $E$ are the vertex and edge sets of the input graph, respectively;
$\out$ is the output size of the query; 
$\out_a$ is the maximal size of the uncalibrated output of any CRPQ appearing in the query; 
and $\out_c$ is the maximal size of the calibrated output of any CRPQ appearing in the query.

\renewcommand{\arraystretch}{1.5}
\begin{center}
\begin{tabular}{@{}c || c | c@{}}
CRPQs & \ospgosyan & Our algorithm \\
\hline
free-con. acyc. & 
$|E| + |E|\hspace{-0.1em}\cdot\hspace{-0.1em} \out_a^{1/2} + \out$ & $|E| + |E|\hspace{-0.1em} \cdot\hspace{-0.1em} \out_c^{1/2} + \out$ \\
\hline
general acyc. & 
$|E| + |E|\hspace{-0.1em}\cdot\hspace{-0.1em} \out_a^{1/2} + \out_a \hspace{-0.1em} \cdot\hspace{-0.1em} \out^{1-1/\fw} + \out$ & $|E| + |E|\hspace{-0.1em}\cdot \hspace{-0.1em}\out^{1/2} \cdot |V|^{\ctw/2}+ \out |V|^{\ctw}$
\end{tabular}
\end{center}

We start the comparison with a free-connex acyclic query. On this type of queries, our algorithm performs at least as good as \ospgosyan and outperforms \ospgosyan when $\out_a$ is asymptotically larger than $\out_c$.   

\begin{example}
\label{ex:analysis_running_time_free_connex}
\rm
Consider the free-connex acyclic CRPQ $Q(X,Y,Z) = a^*aa(X,Y) \wedge b^*bb(Y,Z)$, which asks for all triples $(u,v,w)$ of vertices in the input graph such that there is a path from $u$ to $v$ labeled by $a^*aa$ and a path from $v$ to $w$ labeled by $b^*bb$.
The free-connex fractional hypertree width $\fw(Q)$ of $Q$ is $1$, while its contraction width $\ctw(Q)$ is 0.

We evaluate $Q$ over the edge-labeled graph $G = (V_1 \cup V_2\cup \{v\}, E_a \cup E_b, \{a,b\})$, where $V_1 = \{u_1, \ldots , u_n\}$, $V_2 = \{w_1, \ldots , w_n\}$, $E_a = \{(u_i,a,v), (v,a,w_i) \mid i \in [n]\}$, and $E_b = \{(u_i,b,v), (v,b,w_i) \mid i \in [n]\}$, for some $n \in \mathbb{N}$. The graph $G$ is depicted in Figure~\ref{fig:graphs_part1} (left). 

We have $|V| =2n+1$ and $|E| = 4n$.
The output of each of the RPQs $a^*aa$ and $b^*bb$ is $\{(u_i, w_j)\mid i, j \in [n]\}$, so of size $n^2$. The output of $Q$ and therefore the calibrated outputs of the RPQs are empty.  
Hence, $\out = \out_c = 0$ and $\out_a = n^2$.
Plugging these quantities into the complexity expressions of the algorithms, we observe that \ospgosyan takes $\bigO{n^2}$ time, whereas our approach takes only $\bigO{n}$ time.
\qed
\end{example}

\begin{figure}[t]
\centering
\begin{tikzpicture}
  \node (u1) at (0,0) {$u_{1}$};
  \node (udots) at (0,-1) {$\vdots$};
  \node (un) at (0,-2) {$u_{n}$};

  \node (v) at (1.5,-1) {$v$};

  \node (w1) at (3,0) {$w_{1}$};
  \node (wdots) at (3,-1) {$\vdots$};
  \node (wn) at (3,-2) {$w_{n}$};

  \draw[->] (u1) to[bend left=10] node[above] {$a$} (v);
  \draw[->] (u1) to[bend right=10] node[pos=0.3, below] {$b$} (v);
  \draw[->] (un) to[bend left=10] node[above] {$a$} (v);
  \draw[->] (un) to[bend right=10] node[pos=0.4, below] {$b$} (v);

  \draw[->] (v) to[bend left=10] node[above] {$a$} (w1);
  \draw[->] (v) to[bend right=10] node[pos=0.75, below] {$b$} (w1);
  \draw[->] (v) to[bend left=10] node[above] {$a$} (wn);
  \draw[->] (v) to[bend right=10] node[below] {$b$} (wn);

\begin{scope}[xshift =5.5cm]
  \node (u0) at (0,0) {$u_{0}$};
    \node (v0) at (1.5,0) {$v_{0}$};
        
  \node (u1) at (0,-0.5) {$u_{1}$};
  \node (udots) at (0,-1.1) {$\vdots$};
  \node (un) at (0,-2) {$u_{n}$};

  \node (v) at (1.5,-1) {$v$};

  \node (w1) at (3,-0.5) {$w_{1}$};
  \node (wdots) at (3,-1.1) {$\vdots$};
  \node (wn) at (3,-2) {$w_{n}$};

  \node (z1) at (4,0) {$z_{1}$};
  \node (zdots) at (4,-0.5) {$\vdots$};    
  \node (zn) at (4,-1.2) {$z_{n}$};    

  \draw[<-] (u1) to node[above] {$a$} (v);
  \draw[<-] (un) to node[above] {$a$} (v);

  \draw[<-] (v) to  node[above] {$a$} (w1);
  \draw[<-] (v) to node[below] {$a$} (wn);

    \draw[->] (u0) to node[above] {$a$} (v0);
    \draw[->] (v0) to node[above] {$a$} (w1);

    \draw[->] (w1) to node[above] {$b$} (z1);        
    \draw[->] (w1) to node[below] {$b$} (zn);                
\end{scope}

\end{tikzpicture}
\caption{The edge-labeled graphs from Examples~\ref{ex:analysis_running_time_free_connex} (left) and \ref{ex:analysis_running_time_k_path} (right).}
\label{fig:graphs_part1}
\end{figure}
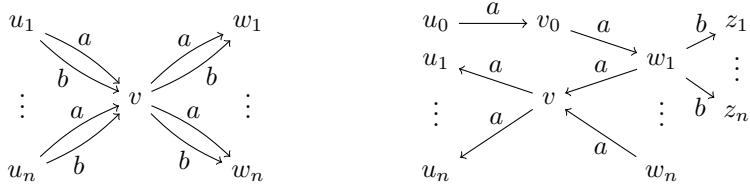

\begin{figure}[t]
\centering
\begin{tikzpicture}
  \node (u0) at (0,0) {$u_{0}$};
    \node (v0) at (1.5,0) {$v_{0}$};
    
    \node (z1) at (4,0) {$z_{1}$};
    \node (z2) at (4,-1) {$z_{2}$};    
    
  \node (u1) at (0,-0.5) {$u_{1}$};
  \node (udots) at (0,-1.1) {$\vdots$};
  \node (un) at (0,-2) {$u_{n}$};

  \node (v) at (1.5,-1) {$v$};

  \node (w1) at (3,-0.5) {$w_{1}$};
  \node (wdots) at (3,-1.1) {$\vdots$};
  \node (wn) at (3,-2) {$w_{n}$};

  \draw[<-] (u1) to node[above] {$a$} (v);
  \draw[<-] (un) to node[above] {$a$} (v);

  \draw[<-] (v) to  node[above] {$a$} (w1);
  \draw[<-] (v) to node[below] {$a$} (wn);

    \draw[->] (u0) to node[above] {$a$} (v0);
        \draw[->] (v0) to node[above] {$a$} (w1);

        \draw[->] (z1) to node[above] {$b$} (w1);        
        \draw[->] (z2) to node[below] {$c$} (w1);                

\begin{scope}[xshift = 6cm]
  \node (u1) at (0,-1) {$u_{1}$};
  \node (udots) at (0,-1.4) {$\vdots$};
  \node (un) at (0,-2) {$u_{n}$};

  \node (w1) at (3,-1) {$w_{1}$};
  \node (wdots) at (3,-1.4) {$\vdots$};
  \node (wn) at (3,-2) {$w_{n}$};

  \node (z1) at (1,0) {$z_{1}$};
  \node (ldots) at (1.5,0) {$\ldots$};
  \node (zn) at (2,0) {$z_{n}$};  

  \node (v) at (1.5,-1) {$v$};
  
  \draw[->] (u1) to node[pos=0.3, above] {$a$} (v);
  \draw[->] (un) to node[pos=0.5, below] {$a$} (v);

  \draw[<-] (v) to  node[pos=0.75, above] {$b$} (w1);
  \draw[<-] (v) to node[pos=0.7, below] {$b$} (wn);

    \draw[->] (z1) to  node[left] {$c$} (v);
  \draw[->] (zn) to node[right] {$c$} (v);

\end{scope}
\end{tikzpicture}
\caption{The edge-labeled graphs from Examples~\ref{ex:analysis_running_time_non_free_connex} (left) and \ref{ex:analysis_running_time_bad_case} (right).}
\label{fig:graphs_part2}
\end{figure}
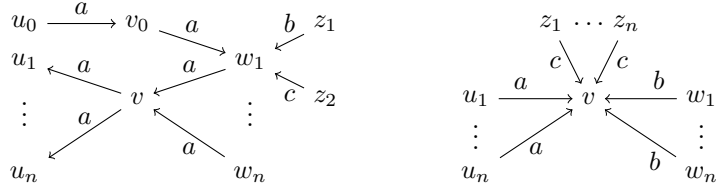

The next example illustrates a case where our approach outperforms \ospgosyan on a CRPQ that is not free-connex acyclic but has contraction width 0.

\begin{example}
\label{ex:analysis_running_time_k_path}
\rm
Consider the 2-path CRPQ $Q(X, Z) = a^*aa(X, Y) \wedge b(Y, Z)$ which asks for all pairs of vertices $(u, w)$ in the input graph such that there are a vertex $v$, a path from $u$ to $v$ labeled with $a^*aa$ and an edge from $v$ to $w$ labeled with $b$. This query is acyclic, but not free-connex. It has free-connex fractional hypertree width 2 and contraction width 0.

Consider the edge-labeled graph $G = (V_1 \cup V_2 \cup V_3 \cup \{v, v_0\}, E_a \cup E_b, \{a, b\})$, where $V_1 = \{u_0, \ldots, u_n\}$, $V_2 = \{w_1, \ldots, w_n\}$, $V_3 = \{z_1, \ldots, z_n\}$, $E_a = \{(w_i,a,v), (v,a,u_i) \mid i \in [n]\} \cup \{(u_0, a, v_0), (v_0, a, w_1)\}$, and $E_b = \{(w_1, b, z_i) \mid i \in [n]\}$ for some $n \in \mathbb{N}$. This graph is visualized in Figure~\ref{fig:graphs_part1} (right).

We have $|V| = 3n+3$, $|E| = 3n+2$, and $\out_a = \Theta(n^2)$. 
The output of $Q$ consists of the tuples $\{(u_0, z_i) \mid i \in [n]\}$, hence $\out = n$. 
Our algorithm needs $\bigO{n^{3/2}}$ time to evaluate this query
on the input graph, whereas \ospgosyan approach requires $\bigO{n^{5/2}}$ time.
\qed
\end{example}

The next example illustrates a case where our algorithm outperforms \ospgosyan on a query whose contraction width is greater than 0.

\begin{example}
 \label{ex:analysis_running_time_non_free_connex}
\rm
Consider the 3-star CRPQ $Q(X_1,X_2, X_3) = a^*aa(X_1,X) \wedge b(X_2,X)\wedge c(X_3,X)$, which asks for all triples $(u,v,w)$ of vertices such that there is a vertex $z$, a path from $u$ to $z$ labeled by $a^*aa$, an edge from $v$ to $z$ labeled by $b$, and an edge from $w$ to $z$ labeled by $c$.
The query is acyclic but not free-connex. It has free-connex fractional hypertree width $3$ and contraction width $1$.

Consider the edge-labeled graph $G = (V_1 \cup V_2 \cup \{u_0, v_0, v, z_1, z_2\}, E_a \cup E_{bc}, \{a,b,c\})$, where $V_1 = \{u_1, \ldots , u_n\}$, $V_2 = \{w_1, \ldots , w_n\}$, 
$E_a = \{(w_i,a,v), (v,a,u_i) \mid i \in [n]\}$ $\cup$
$\{(u_0, a, v_0), (v_0, a, w_1)\}$, and $E_{bc} = \{(z_1,b,w_1), (z_2,c,w_1)\}$ 
for some $n \in \mathbb{N}$. The graph $G$ is visualized in Figure~\ref{fig:graphs_part2} (left). 

We observe that  $|V| =2n+5$, $|E| = 2n + 4$, and $\out_a = \Theta(n^2)$.
The output of $Q$ consists of the single tuple $(u_0,z_1, z_2)$, hence $\out = 1$.
These quantities imply that \ospgosyan needs $\bigO{n^2}$, while our algorithm needs only 
$\bigO{n^{3/2}}$ time. 
\qed
\end{example}

Next, we illustrate a CRPQ and an input graph where \ospgosyan outperforms our algorithm. This CRPQ degenerates to a conjunctive query.  
Given the size of the outputs of the RPQs, the final output of the CRPQ achieves asymptotically its worst-case size.
\begin{example}
\label{ex:analysis_running_time_bad_case}
\rm
Consider the following simplified variant of the 3-star CRPQ 
from Example \ref{ex:analysis_running_time_non_free_connex}: 
$Q(X_1,X_2, X_3) = a(X_1,X) \wedge b(X_2,X)\wedge c(X_3,X)$.
The query asks for all vertex triples $(u,v,w)$ 
such that there is a vertex $z$, an edge from $u$ to $z$ labeled by $a$, an edge from $v$ to $z$ labeled by $b$, and an edge from $w$ to $z$ labeled by $c$. Just as in Example~\ref{ex:analysis_running_time_non_free_connex}, the query has free-connex fractional hypertree width $3$ and contraction width $1$.

We evaluate the query over the edge-labeled graph 
$G = (V_1 \cup V_2 \cup V_3 \cup \{v\}, E, \{a,b,c\})$, where $V_1 = \{u_1, \ldots , u_n\}$, $V_2 = \{w_1, \ldots , w_n\}$,
$V_3 = \{z_1, \ldots , z_n\}$, and 
$E$ $=$ $\{(u_i,a,v),$ $(w_i,b,v),$ $(z_i,c,v)$ $\mid i \in [n]\}$
for some $n \in \mathbb{N}$. The graph $G$ is shown in Figure~\ref{fig:graphs_part2} (right). 

We have $|V| =3n+1$ and $|E| = 3n$.
The outputs of the RPQs $a$, $b$, and $c$ are 
$\{(u_i, v) \mid i \in [n]\}$, $\{(w_i, v) \mid i \in [n]\}$, and 
$\{(z_i, v) \mid i \in [n]\}$, respectively. 
The output of $Q$ is 
$\{(u_i,v_j, z_k) \mid i,j,k \in [n]\}$.
Hence, $\out_a = n$ and $\out = n^3$.
%
We observe that \ospgosyan
runs in $\bigO{n^3}$ time, while our algorithms requires $\bigO{n^4}$ time.
\qed
\end{example}

\section{Missing Details in Section \ref{sec:prelims}}
\label{app:prelims}

\paragraph*{Complexity Measures}
\nop{The {\em evaluation problem} for a CRPQ $Q$ has the following input and output:

\smallskip
    \begin{tabular}{ll} 
        \textbf{Input} & An edge-labeled graph $G = (V,E, \Sigma)$\\
        \textbf{Output} & The set of all tuples in the output of $Q$ over $G$
    \end{tabular}

\smallskip 
}
We use data complexity, i.e., we consider the query fixed and of constant size and measure the time complexity of solving the evaluation problem for CRPQs in terms of the number $|V|$ of vertices and the number $|E|$ of  edges of the input graph, the  size $\out$ of the output of the query, and the maximum size $\out_c$ of the calibrated output of any RPQ in the query.

We consider the RAM model of computation and assume that each materialized relation over the vertex set of the input graph is implemented by a data structure that allows to look up a tuple in
the relation in constant time and enumerate all tuples with constant delay.

\paragraph*{Conjunctive Queries}
A conjunctive query (CQ)   is of the form 
\begin{align*}
    Q(\bm F) = R_1(\bm X_1) \wedge \cdots \wedge R_n(\bm X_n),
    \label{eq:join-query}
\end{align*}
where 
each $R_i$ is a relation symbol, each 
$\bm X_i$ is a tuple of variables, and 
$\bm F \subseteq  \bigcup_{i \in [n]} \bm X_i$
is the set of free variables.
We refer to $\bm X_i$ as the {\em schema}
of $R_i$ and
call each $R_i(\bm X_i)$ an {\em atom} of $Q$.

The domain of a variable $X$ is denoted by  $\Dom(X)$. 
A {value tuple} $\bm t$ over a schema $\bm X = (X_1, \ldots , X_k)$ is an element from $\Dom(\bm X) = \Dom(X_1) \times \cdots \times \Dom(X_k)$.
We denote by $\bm t(X)$ the $X$-value of $\bm t$. 
A tuple $\bm t$ over $\bm F$
is in the result of $Q$ if it can be extended into a tuple $\bm t'$ over $\bigcup_{i \in [n]} \bm X_i$ such that there are tuples 
$\bm t_i \in R_i$ for $i \in [n]$ and the following holds: for all $X \in \vars(Q)$, the set $\{\bm t_i(X)\mid i \in [n] \wedge X \in \bm X_i\}$ consists of the single value $\bm t'(X)$.

A CQ is {\em ($\alpha$-)acyclic} if we can construct a tree, called {\em join tree}, such that the nodes of the tree are the atoms of the query ({\em coverage}) and for each variable, it holds: if the variable appears in two atoms, then it appears in all atoms on the path connecting the two atoms ({\em connectivity})~\cite{Brault-Baron16}.

\paragraph*{Queries Combining RPQs with Unary and Nullary Relations}
Our evaluation algorithms use at intermediate steps  
queries of the form 
$$Q(\bm F) = r(X,Y) \wedge \bigwedge_{i \in [k]} R_i(\bm X_i),$$
where $r$ is an  RPQ and each $R_i$ is a unary or nullary materialized  relation over the vertex set of the input graph.
We have $\bm F \subseteq \vars(Q)$, where 
$\vars(Q): = \{X,Y\} \cup \bigcup_{i \in [k]} \bm X_i$. 

The semantics of such a query over an edge-labeled graph $G = (V, E, \Sigma)$ is straightforward. 
A mapping $\mu: \bm F \rightarrow V$ is in the output of $Q$ if it can be extended to a mapping
$\mu': \vars(Q) \rightarrow V$
such that $G$ has a path from $\mu'(X)$
to $\mu'(Y)$ labeled by a string from $L(r)$ and for each $i \in [k]$,
it holds: if $\bm X_i = Z$ for some variable $Z$, then  $\mu'(Z) \in R_i$; if $\bm X_i$ is empty, then  $R_i$ consists of the empty tuple (which means that it evaluates to \texttt{true}). 
We represent the mapping $\mu$ by the tuple $(\mu(X))_{X \in \bm F}$, assuming a fixed order on the free variables.
\section{Missing Details in Section~\ref{sec:contraction}}
\label{app:contraction}
\subsection{Proof of Proposition~\ref{lem:unique-reduced-graph}}

\UniqueReducedGraph*

\begin{proof}
    The proof is by induction on the number of vertices in the query graph $G$ of the acyclic CRPQ $Q$.
    In the base case, the graph $G$ has no bound vertices of degree less than $3$, in which case the only possible elimination order is the empty order $\omega=()$, which yields $G$ itself.

    For the induction step, suppose that $G$ has at least one bound vertex $v$ of degree less than $3$.
    If $G$ has exactly one bound vertex $v$ of degree less than $3$, then every elimination order $\omega$ of $G$ must start with $v$.
    Let $G_1:=\eliminate(G, v)$.
    By induction, $G_1$ must have a unique contracted graph $G_k$, hence $G_k$ must also be the unique contracted graph of $G$.
    Now suppose that $G$ has at least two bound vertices $v, v'$ of degree less than $3$,
    and let $G_1:=\eliminate(G, v)$ and $G_1':=\eliminate(G, v')$. 
    By induction, $G_1$ has a unique contracted graph $G_k$ and $G_1'$ has a unique contracted graph $G'_{k'}$.
    We make the following claim:
    \begin{claim}
        \label{clm:all-roads-lead-to-rome}
        $\eliminate(G_1, v') = \eliminate(G_1', v)$
    \end{claim}
    Assuming Claim~\ref{clm:all-roads-lead-to-rome} holds, $G_k$ must be the unique contracted graph of
    $\eliminate(G_1, v')$ and $G'_{k'}$ must be the unique contracted graph of $\eliminate(G_1', v)$.
    Therefore, $G_k = G'_{k'}$ is the unique contracted graph of both $G_1$ and $G_1'$, and by extension $G$, which concludes the proof.

    Now we prove Claim~\ref{clm:all-roads-lead-to-rome}.
    \begin{proof}[Proof of Claim~\ref{clm:all-roads-lead-to-rome}]
        We recognize the following cases:
        \begin{itemize}
            \item Case 1: $v$ and $v'$ are not adjacent in $G$. In this case, eliminating $v$ does not affect the degree of $v'$, and vice versa. If both degrees are less than 3, then the two vertices
            can be eliminated independently from one another, leading to the same graph, regardless
            of which vertex is eliminated first. Therefore, $\eliminate(G_1, v') = \eliminate(G_1', v)$.
            \item Case 2: $v$ and $v'$ are adjacent in $G$. Here, we recognize the following sub-cases:
            \begin{itemize}
                \item Case 2a: Both $v$ and $v'$ have degree 1 in $G$.
                In particular, $v$ and $v'$ are connected by an edge and have no other neighbors.
                In this case, eliminating either vertex
                will result in removing both vertices along with the edge connecting them.
                \item Case 2b: Both $v$ and $v'$ have degree 2 in $G$. Namely, $v$ has two neighbors $u$ and $v'$, whereas $v'$ has two neighbors $v$ and $u'$ where $u\neq u'$.
                In this case, eliminating both $v$ and $v'$ (in any order) will result in removing both vertices along with the edge connecting them, and adding a new edge between $u$ and $u'$.
                \item Case 2c: $v$ has degree 1 and $v'$ has degree 2 in $G$. In particular, $v'$ is the only neighbor of $v$, and $v'$ has another neighbor $u$.
                In this case, eliminating $v$ and $v'$ (in any order) will result in removing both vertices along with the two edges $\{v, v'\}$ and $\{v', u\}$.
            \end{itemize}
        \end{itemize}
    \end{proof}
\end{proof}

\subsection{Fractional Hypertree Width}
\label{sec:fract_hypertree_width}
We introduce the {\em fractional edge cover number} of variable sets in    CRPQs.
Given a CRPQ $Q$ and $\bm F \subseteq \vars(Q)$,    
a {\em fractional edge cover}
of $\bm F$ is a solution 
$\boldsymbol{\lambda} = (\lambda_{r(\bm X)})_{r(\bm X) \in \atoms(Q)}$ to the following linear program \cite{AtseriasGM13}: 
\begin{align*}
\text{minimize} & \TAB\sum_{r(\bm X) \in\, \atoms(Q)} \lambda_{r(\bm X)} && \\[3pt]
\text{subject to} &\TAB \sum_{r(\bm X): \, X \in \bm X} \lambda_{r(\bm X)} \geq 1 && \text{ for all } X \in \bm F \text{ and } \\[3pt]
& \TAB\lambda_{r(\bm X)} \in [0,1] && \text{ for all } r(\bm X) \in \atoms(Q)
\end{align*}
The optimal objective value of the above program 
is called the fractional edge cover number of $\bm F$ 
in $Q$ and is denoted as $\rho_{Q}^{\ast}(\bm F)$.  

\nop{
\begin{definition}[The restriction $Q|_{\bm Y}$ of a CRPQ $Q$]
    \label{defn:restriction}
    For a CRPQ $Q$ and a subset $\bm Y \subseteq \vars(Q)$,
    we define the {\em restriction of $Q$ to $\bm Y$}, denoted by $Q|_{\bm Y}$,
to be the query that results
from $Q$ by restricting the schema of each atom to those variables 
that appear in $\bm Y$. Formally:
\begin{align}
    R_{\bm Y}(\bm X \cap \bm Y) &\defeq \pi_{\bm X \cap \bm Y} R(\bm X),\quad\quad
    \text{for all $R(\bm X) \in \atoms(Q)$}
    \nonumber\\
    Q_{\bm Y}(\bm Y) &\defeq \bigwedge_{R(\bm X) \in \atoms(Q)} R_{\bm Y}(\bm X \cap \bm Y)
    \label{eq:bag_query}
\end{align}
\end{definition}
}

\begin{definition}[Fractional Hypertree Width]
    \label{defn:fhtw}
The fractional hypertree width $\fw(Q)$ of any CRPQ $Q$ and
the fractional hypertree  width $\fw(\calT)$ of any tree decomposition $\calT$ of $Q$ are:
$$
\fw(Q)  = \min_{\calT \in \ftd(Q)} \fw(\calT)
\hspace{1cm} \text{and} \hspace{1cm} 
\fw(\calT)  = \max_{\calB \in \bags(\calT)} \rho^*_Q(\calB)
$$
\end{definition}

\section{Missing Details in Section \ref{sec:freeconnex}}
\label{app:freeconnex}
\subsection{Proof of Proposition~\ref{prop:complexity_eval_free_connex}}
\ComplexityEvalFreeConnex*

\begin{proof}
First, we prove that each relation $R_i$ computed in \textsc{EvalFreeConnex} is calibrated. Then, we show that the running time of each of the subprocedures \textsc{BottomUp} and \textsc{TopDown} is $\bigO{|E| + |E| \cdot \out_c^{1/2}}$. Finally, we analyze the running time of the entire procedure \textsc{EvalFreeConnex}.

\paragraph*{All relations $R_i$ are calibrated}
The argumentation is similar to the argumentation  for  the Yannakakis algorithm~\cite{Yannakakis81}. For completeness, we outline the main proof ideas. 
We first show that the relation $R_i$ computed at the root $r_i(X_i, Y_i)$ of the join tree is the calibrated output of $r_i$. For any atom $r_j(X_j, Y_j)$ in the join tree, let $Q_j$ be the full query whose body is the conjunction of all atoms in the subtree rooted at $r_j(X_j,Y_j)$. We prove the following claim:

\begin{claim}
\label{claim:bottom-up-calib}
For each non-root atom $r_i(X_i, Y_i)$, it holds that the relation $B_{i \rightarrow j}(\bm X_i)$ is the projection of the output of $Q_i$ onto $\bm X_i$.
\end{claim}

\begin{proof}
We prove Claim~\ref{claim:bottom-up-calib} using induction  on the structure of the join tree. 
In the base case, the body of $Q_i$ consists of the single atom $r_i(X_i, Y_i)$.
Since $B_{i \rightarrow j}$ is a projection of $r_i(X_i,Y_i)$, Claim~\ref{claim:bottom-up-calib} holds in this case. 

Consider now an atom $r_i(X_i, Y_i)$ in the join tree that is not a leaf. 
Consider the queries $Q_{c_1}, \ldots, Q_{c_k}$ constructed for the child 
atoms $r_{c_1}(X_{c_1},Y_{c_1}), \ldots , r_{c_k}(X_{c_k},Y_{c_k})$.
By the induction hypothesis, each relation $B_{c_j \rightarrow i}(\bm X_{c_j})$ is the projection of the output of $Q_{c_j}$ onto $\bm X_{c_j}$. 
By the definition of $B_{i \rightarrow j}$, the relation $B_{i \rightarrow j}$ is the projection of the output of the following query onto $\bm X_i$:
\[
Q_i'(X_i, Y_i) = r_i(X_i, Y_i) \wedge \bigwedge\limits_{r_c(X_c, Y_c)\in\text{children}(r_i(X_i,Y_i))} B_{c \rightarrow i}(\{X_c, Y_c\} \cap \{X_i, Y_i\}),
\]
Since $\bm X_i \subseteq \{X_i, Y_i\}$, the relation $B_{i \rightarrow j}$ must be the projection of the output of $Q_i$ onto $\bm X_i$.
This completes the proof of Claim~\ref{claim:bottom-up-calib}.
\end{proof}
Claim~\ref{claim:bottom-up-calib} implies that the relation $R_i$ computed at the root of the join tree must be the calibrated output of the root atom $r_i(X_i, Y_i)$.

\smallskip

We now show that each relation $R_i$ computed be the procedure  \textsc{EvalFreeConnex} is the calibrated output of $r_i$. We use induction on the depth of a node $r_i(X_i, Y_i)$ in the join tree. In the base case,  $r_i(X_i, Y_i)$ is the root of the join tree. We have shown above that 
$R_i$ must be the calibrated output of $r_i$.
For the induction step, consider a non-root atom $r_i(X_i, Y_i)$. Let $r_p(X_p, Y_p)$ be the parent atom of $r_i(X_i, Y_i)$. By the induction hypothesis, $R_p$ is calibrated output of $r_p$. Since $T_{p \rightarrow i}$ is a projection of 
$r_p$ onto the common variables of $r_p$ and $r_i$, this implies that $R_i$ is the calibrated output of $r_i$.

We conclude that all relations $R_i$ computed in \textsc{EvalFreeConnex} are calibrated.

\paragraph*{Running time of \textsc{BottomUp}} 
The running time of \textsc{BottomUp} is dominated by the time to materialize the relations $B_{i \rightarrow p}$ and the relation $R_i$ at the root of the join tree. We first show that materializing all relations $B_{i \rightarrow p}$ takes $\bigO{|E|}$ time. Since the query is acyclic, there are no two atoms with the same pair of variables. Therefore, the definition of each relation 
$B_{i \rightarrow p}$ is as follows:
\[
B_{i\rightarrow p}(\bm X) = r_i(X_i, Y_i) \wedge B_{j_1\rightarrow i}(\bm X_1) \wedge \cdots \wedge B_{j_k\rightarrow i}(\bm X_k),
\]
where $\bm X$ and each $\bm X_i$ either consists of $X_i$, or consists of $Y_i$, or is empty. 
Assuming that the relations $B_{j_\ell \rightarrow i}$ 
appearing in the body of the query defining $B_{i \rightarrow p}$ are materialized,  
Lemma~\ref{lem:reachable_set_ospg_variant} implies that 
each relation $B_{i \rightarrow p}$ can be computed in $\bigO{|E|}$ time. Since the query size is constant, and for each query atom $r_i(X_i,Y_i)$, at most one relation $B_{i \rightarrow p}$ is constructed, the time required to compute all relations of type $B_{i \rightarrow p}$ is $\bigO{|E|}$.

As shown above, the relation $R_i$ constructed at the root of the join tree is calibrated. The relation $R_i$ is of the form:
\[
R_i(\{X_i, Y_i\} \cap \bm F) = r_i(X_i, Y_i) \wedge B_{j_1\rightarrow i}(\bm X_1) \wedge \cdots B_{j_k\rightarrow i}(\bm X_k),
\]
where each $\bm X_i$ either consists of $X_i$, or consists of $Y_i$, or is empty. 
By Lemma \ref{lem:reachable_set_ospg_variant}, it holds:
If $\{X_i, Y_i\} \cap \bm F$ is empty or a singleton, the relation $R_i$ can be computed in 
$\bigO{|E|}$ time; if $\{X_i, Y_i\} \cap 
\bm F = \{X_i, Y_i\}$, the computation time is $\bigO{|E| + |E| \cdot \out_{X_i, Y_i}^{1/2}}$ time, where $\out_{X_i, Y_i}$ is the output size of $R_i$. Since $\out_{X_i, Y_i} \leq \out_c$, we can rewrite this time complexity as $\bigO{|E| + |E| \cdot \out_c^{1/2}}$. 

We conclude that the subprocedure \textsc{BottomUp} runs in $\bigO{|E| + |E| \cdot \out_c^{1/2}}$ time.

\paragraph*{Running time of \textsc{TopDown}} 
The running time of \textsc{TopDown} is dominated by the time needed to compute the relations denoted $T_{p \rightarrow i}$ and the relations $R_i$ constructed at the atoms $r_i(X_i,Y_i)$ below the root of the join tree. Consider a node atom $r_i(X_i, Y_i)$ in the join tree 
and its parent node $r_p(X_p, Y_p)$ at which we have already
computed the relation $R_p(\bm X)$. We have  
$T_{p\rightarrow i}(\{X_i, Y_i\} \cap \bm X) = R_p(\bm X)$ and
$R_i(\{X_i, Y_i\}\cap \bm F) = r_i(X_i, Y_i) \wedge  T_{p\rightarrow i}(\{X_i, Y_i\} \cap \bm X)$.
Since the query cannot contain two atoms with the same set of variables, the set $\{X_i, Y_i\} \cap \bm X$
can contain at most one variable. Hence, it follows 
from Lemma~\ref{lem:reachable_set_ospg_variant} 
that $R_i$ can be computed in  $\bigO{|E| + |E| \cdot \out_{X_i, Y_i}^{1/2}}$ time, where
$\out_{X_i, Y_i}$ is the output size of $R_i$. 
Again, as we showed above, the relation $R_i$ is calibrated. This means,
we can express its computation time as 
$\bigO{|E| + |E| \cdot \out_{c}^{1/2}}$.

Each relation $T_{p \rightarrow i}$ is a projection of the relation $R_p$, which means that
$T_{p \rightarrow i}$ can be computed by a single pass over $R_{p}$. Since the size of $R_p$ is upper-bounded
by $\bigO{|E| + |E| \cdot \out_c^{1/2}}$,  
the relation $T_{p \rightarrow i}$ can be computed within the
same complexity bound.

We conclude that the running time of the procedure \textsc{TopDown} is $\bigO{|E| + |E| \cdot \out_c^{1/2}}$.

\paragraph*{Running time of \textsc{EvalFreeConnex}} 
As shown above, each of the procedures  \textsc{BottomUp} and \textsc{TopDown} takes 
$\bigO{|E| + |E| \cdot \out_c^{1/2}}$ time. 
The procedure \textsc{EvalFreeConnex} computes and returns the output of the conjunctive query $Q'(\bm F)$, which joins the calibrated relations $\calR = \{R_{i_1}, \ldots, R_{i_k}\}$. 
Since all variables of the query  $Q'$ are free and the relations in $\calR$
correspond to nodes in a connected subtree of a join tree, the query $Q'$ forms a free-connex acyclic conjunctive query. Using the Yannakakis algorithm, the output of $Q'$ 
can be computed in 
$\bigO{\max_{j \in [k]} |R_{i_j}| + \out}$ time. 
Since the size of the  relations $R_{i_j}$ is 
$\bigO{|E| + |E| \cdot \out_c^{1/2}}$, the overall computation time of \textsc{EvalFreeConnex} can be expressed as $\bigO{|E| + |E| \cdot \out_c^{1/2} + \out}$.
\end{proof}

\subsection{Proof of Lemma~\ref{lem:reachable_set_ospg_variant}}
\ReachableSetOSPGVariant*

Before giving the proof of the lemma, we recall some standard concepts. 
A {\em (non-deterministic) finite automaton} $\calA = (\calS, s_0, \Sigma, \delta, \calF)$ consists of a set $\calS$ of states, 
a start state $s_0$, an input alphabet $\Sigma$, a transition relation 
$\delta \subseteq \calS \times \Sigma \times \calS$, and a set $\calF$
of accepting states. The language $L(A)$ recognized by a finite automaton 
is defined in the standard way. 
A language is regular if and only if it is recognized by a finite automaton~\cite{HopcroftMU07}. 

Given an edge-labeled graph $G=(V, E, \Sigma)$ and a finite automaton 
$\calA = (\calS, s_0, \Sigma, \delta, \calF)$,  the {\em product graph} of $G$ and $\calA$ is a directed 
graph $P_{G,\calA} = (V_{G,\calA}, E_{G,\calA})$, where
$V_{G,\calA} = V \times \calS$ and 
$E_{G,\calA} = \{((u_1, s_1), (u_2, s_2)) \mid \exists a \in \Sigma \text{ with } 
(u_1, a, u_2) \in V \text { and } (s_1, a, s_2) \in \delta\}$. 
The graph $G$ contains a path from vertex $u$ to vertex $v$  labeled by a string from $L(r)$ for some regular expression $r$ if and only if the product graph of $G$ and the finite automaton 
recognizing $L(r)$ contains a path from $(u,s_0)$ to $(v,f)$ for $f \in \calF$ and some $u,v \in V$.

\begin{proof}[Proof of Lemma~\ref{lem:reachable_set_ospg_variant}]
Let $G =(V,E,\Sigma)$ be an edge-labeled graph.
Consider the relation $U$ defined by
$U(X) = \bigwedge_{i \in [k]} U_i(X)$, which we can compute as follows. 
We iterate over the vertices in $U_1$ and for each such vertex $u$, we do a constant-time lookup in each of the other relations $U_2, \ldots , U_k$. 
If $u$ is contained in all relations $U_i$, we add it to $U$. 
This computation requires $\bigO{|V|}$ time. Due to our assumption that the input graph does not have any isolated vertex (see Section~\ref{sec:prelims}), the overall computation time is $\bigO{|E|}$. Similarly, we can compute the relations 
$W(Y) = \bigwedge_{i \in [\ell]} W_i(Y)$ and 
$Z() = \bigwedge_{i \in [m]} Z_i()$ in $\bigO{|E|}$ time. 
Then, we rewrite $Q$ as 
$Q(\bm X) = r(X, Y) \wedge U(X) \wedge W(Y)$, by assuming, without loss of generality, that $Z()$ evaluates to \texttt{true}. 
If $Z()$ evaluates to \texttt{false}, we can immediately report that the output of $Q$ is empty.

We consider the case that $\bm X$ consists of only $X$ or only $Y$.
In this case, we construct the product 
of the input graph $G$ and the finite automaton recognizing $L(r)$. 
Then, we reverse the edges in the product graph. Starting from the vertices in 
$\{(v,s) \mid v \in W \text{ and } $s$ \text{ is an accepting state}\}$ and
using a (multi-source) graph search algorithm, such as breadth-first search, we find all vertices 
in the set $\{(u,s) \mid u \in U \text{ and } $s$ \text{ is the start state}\}$.
This algorithm runs in time linear in the number of edges in the product graph. Since the size of the query, and hence the size of the finite automaton, is considered constant, the overall computation time is 
$\bigO{|E|}$.

In case the set $\bm X$ is empty, we use the same construction and follow the same search strategy as above. The only difference is that we stop
as soon as we find one vertex in the product graph that belongs to  
$\{(u,s) \mid u \in U \text{ and } $s$ \text{ is the start state}\}$.
In this case, we report that the output of the query is \texttt{true}.
Otherwise, we report that the output is \texttt{false}.
Hence, the computation time is $\bigO{|E|}$, as in the case above. 

Now, we consider the case that $\bm X$ consists of $X$ and $Y$. We use the \ospg algorithm~\cite{KhamisKOS25} on a slight modification of the RPQ $r$ and the input graph $G$ to compute the output of $Q$. 
Consider the RPQ $r' = \sigma_X \cdot r \cdot \sigma_Y$, where $\sigma_X$ and $\sigma_Y$ are fresh symbols that do not occur in $r$ or in $G$.
Let $G'$ be the edge-labeled graph obtained from $G$ by adding a self-loop edge 
$(u, \sigma_X, u)$ for each $u \in U$ and a self-loop edge $(v, \sigma_Y, v)$ for each $v \in W$. The graph $G'$ can be constructed in $\bigO{|E|}$ time by iterating over the vertices in $U$ and $W$ and adapting the corresponding vertices in $G$.
A tuple $(u,v)$ is contained in the output of $Q$ evaluated over $G$
if and only if it is contained in the output of the RPQ $r'$ evaluated over $G'$.
We evaluate $r'$ over $G'$ using the algorithm \ospg. This algorithm runs in $\bigO{|E| + |E| \cdot \out^{1/2}}$ time, where $\out$ is the output size of $r'$, and hence the output size of $Q$.
\nop{
Therefore, each pair of values $(x, y)$ in the original data graph satisfying $R$ corresponds to a pair of nodes $(x', y')$ in the modified data graph satisfying the RPQ $\sigma_x \cdot r \cdot \sigma_y$. As a result, computing $R$ is equivalent to evaluating the RPQ $\sigma_x \cdot r \cdot \sigma_y$ on the modified data graph. We evaluate this RPQ using \ospag. The running time of the \ospg algorithm depends not on the output size of $r$, but rather on the output size of the RPQ $\sigma_x \cdot r \cdot \sigma_y$, which is exactly the output size of $R$. As a result, we can compute $R$ in time $O(|E| + |E| \cdot |R|^{1/2})$.
}
\end{proof}
\section{Missing Details in Section~\ref{sec:general}}
\label{app:general}
\subsection{Proof of Proposition~\ref{prop:contract_query}}
\ContractQuery*

Consider an acyclic CRPQ $Q$ and an edge-labeled graph $G = (V,E, \Sigma)$.
Let $Q_0:=Q$, $G_0:= G$, and $\calR_0:= \{R_X := V \mid X \in \vars(Q_0)\}$ be the set of filter relations constructed in Line~3 of the procedure \textsc{Contract}.
Let $\omega = (X_1, \ldots , X_k)$ be an elimination order for $G_0$ with 
$\omega(G_0) = (G_1, \ldots, G_k)$.
We denote by $\calR_i$ the set of filter relations obtained after the elimination step for variable $X_i$. 
Let $Q_1, \ldots ,Q_k$ be the queries 
constructed by the procedure \textsc{Contract} in Lines 7 and 10 after each elimination step. Let $\hat{G}'$ be the graph constructed from $G$ in Lines 11 and 12 of the procedure \textsc{Contract}. For a CRPQ $Q$ and a set $\calR$ of filter relations, we use $Q \wedge \calR$ as a shorthand for the query 
obtained from $Q$ by adding to the body of $Q$ the conjunction 
$\bigwedge_{R_X \in \calR} R_X(X)$ of the unary relations in $\calR$.

Our proof consists of three parts. First, we show that the following claim holds:

\begin{claim}
\label{claim:contract_same_output}
For each $i \in \{0, \ldots, k \}$, the output of 
$Q_{i} \wedge \calR_i$ over $\hat{G}'$ is the same
as the output of $Q$ over $G$.
\end{claim}

\noindent
Then, we show that the output of the query $Q_k'$, as defined in Line~13 of the procedure \textsc{Contract}, over $\hat{G}'$ is the same as  the output of $Q_k \wedge \calR_k$ over the same graph. In the final part, we show that the procedure \textsc{Contract} runs in 
$\bigO{|E|}$ time.

\medskip

We show by induction on $i \in \{0, \ldots, k \}$ that Claim 
\ref{claim:contract_same_output} holds.

In the base case, $i = 0$, we have $Q_i = Q$ and $\calR_i = \calR_0$. Claim~\ref{claim:contract_same_output} trivially holds since: (1) $Q \wedge \calR_0$ is equivalent to $Q$, 
(2) $\hat{G}'$ has exactly the same vertices as 
$G$, (3) every edge in $G$ is contained in $\hat{G}'$, and (4) any edge contained in $\hat{G}'$ but not in $G$ has a label that does not appear in $Q$, which means that the edges in   $\hat{G}'$
that do not appear in $G$ do not affect the evaluation of $Q$.

Our induction hypothesis is that  Claim~\ref{claim:contract_same_output} holds for $i-1$ with 
$i \in \{1, \ldots , k\}$. In the induction step, we show that it holds for $i$.

First, we consider the case that $Q_{i}$ results from $Q_{i-1}$ via the elimination of a variable $X_i$ that has degree 1 in $G_{i-1}$ (Lines~5-8 in \textsc{Contract}). Let $r(\bm X)$ be the atom that contains 
$X_i$ and some other variable $Y$. The procedure  
removes from $Q_{i-1}$ the atom $r(\bm X)$ and  updates the filter 
$R_Y$ by computing: $R_Y'(Y):= R_Y(Y)$ and $R_Y(Y): = r(\bm X) \wedge R_Y'(Y) \wedge R_{X_i}(X_i)$. The filter set $\calR_i$ results from $\calR_{i-1}$ by replacing the relation $R_Y$ with its updated version. 
It follows from the induction hypothesis and the definition 
of $R_Y(Y)$ that $R_Y$ consists of all vertices from which there is a 
path to a vertex contained in $R_{X_i}$ labeled by a string from $L(r)$.
This implies that  
$Q_i \wedge \calR_i$ over $\hat{G}'$ has the same output as $Q$ over $G$.

Now, we consider the case that $Q_{i}$ results from $Q_{i-1}$ via the elimination of a variable $X_i$ that has degree 2 in $G_{i-1}$ (Lines~9-10 in \textsc{Contract}). 
Without loss of generality, we assume that $Q_{i-1}$ contains two atoms $r_1(Y,X_i)$ and $r_2(Z,X_i)$ for distinct variables $Y$ and $Z$. The other cases are handled  completely analogously. 
The procedure  \textsc{Contract} replaces in $Q_{i-1}$ the atoms
$r_1(Y,X_i)$ and $r_2(Z,X_i)$ with the new atom $r_1\sigma_{X_i}\hat{r}_2^R(Y,Z)$. 
By construction, $\hat{G}'$
contains a path from a vertex $u$ to a vertex $v$ labeled by $r_1\sigma_{X_i}\hat{r}_2^R$  if and only if it  contains a path from $u$ to a vertex $w$ contained in $R_{X_i}$ labeled by $r_1$ and a path from $v$ to  $w$ labeled by $r_2$. 
Note that $\calR_i = \calR_{i-1}$.
Hence, it follows from our induction hypothesis 
and the definition of 
$r_1\sigma_{X_i}\hat{r}_2^R$ that 
$Q_i \wedge \calR_i$ over $\hat{G}'$ has the same output as $Q$ over $G$.
This concludes the proof of Claim~\ref{claim:contract_same_output}.

Next, we show that the output of $Q_k'$ over $\hat{G}'$ is the same as the output of $Q_k \wedge \calR_k$ over $\hat{G}'$. The query $Q_k'$ results from $Q_k$ by adding the atom $\sigma_X(X,\bar{X})$ for each variable appearing in $Q_k$, where $\bar{X}$ is a  fresh variable. A tuple $\bm t$ is in the  output of $Q_k'$ if and only if it is in the output of $Q_k$ and each vertex mapped to variable $X$ is contained in the filter $R_X$. Hence, 
$Q_k'$ has the same output over $\hat{G}'$ as $Q_k \wedge \calR_k$.

Lastly, we show that the procedure \textsc{Contract} runs in $\bigO{|E|}$ time. 
The query transformations in Lines~7, 10, and 13 can be done in constant time.
To initialize the filter relations in Line~3, we need to iterate over all vertices in the input graph. 
Due to our assumption that the input graph does not contain isolated vertices (Section~\ref{sec:prelims}), we can find all vertices by iterating over the edges in the input graph, which requires $\bigO{|E|}$ time. 
By Lemma~\ref{lem:reachable_set_ospg_variant}, 
the filter relations can be updated in $\bigO{|E|}$ time. We compute the symmetric closure $\hat{G}$ of $G$ by iterating over the edges in $G$ and adding for each edge $(u,\sigma,v)$, an edge $(v,\hat{\sigma},u)$. Likewise $\hat{G}'$ can be computed by iterating over the edges in $\hat{G}$ and adding a self-loop $(u,\sigma_X,u)$ for each vertex $u$ such that $u$ is contained in the filter relation $R_X$. Hence, $\hat{G}'$ can be computed in $\bigO{|E|}$ time. 
We conclude that the procedure \textsc{Contract} runs in $\bigO{|E|}$ time.

\end{document}